\documentclass[11pt]{article}

\usepackage{times}

\usepackage[paperwidth=199.8mm,paperheight=297mm,centering,hmargin=25mm,vmargin=25mm]{geometry}
\usepackage{authblk} 
\usepackage[titletoc,title]{appendix}
\usepackage{hyperref}
\hypersetup{colorlinks=true,citecolor=blue,linkcolor=blue,filecolor=blue,urlcolor=blue,breaklinks=true}
\usepackage[dvipsnames]{xcolor}
\usepackage{color}
\usepackage[marginal]{footmisc}
\usepackage{url}

\usepackage{mathtools}
\usepackage{amsmath}
\usepackage{amssymb}
\usepackage[shortlabels]{enumitem}
\usepackage{graphicx,epic,eepic,epsfig,latexsym,verbatim}
\usepackage{dsfont}

\usepackage{framed}
\definecolor{shadecolor}{rgb}{0.9,0.9,0.9}
\usepackage{subcaption}
\usepackage{caption}
\usepackage{float}
\usepackage{tikz}
\usepackage{tcolorbox}
\usepackage{relsize}

\definecolor{LightGray}{RGB}{220,220,220}
\definecolor{myred}{RGB}{236, 17, 0}
\definecolor{myblue}{RGB}{10, 88, 153}
\definecolor{myorange}{RGB}{236, 137, 0}
\definecolor{mygreen}{RGB}{26, 152, 81}

\usepackage{amsthm}
\newtheorem{definition}{Definition}

\newtheorem{lemma}[definition]{Lemma}
\newtheorem{theorem}[definition]{Theorem}

\newtheorem{assumption}[definition]{Assumption}
\def\squareforqed{\hbox{\rlap{$\sqcap$}$\sqcup$}}
\def\qed{\ifmmode\squareforqed\else{\unskip\nobreak\hfil
\penalty50\hskip1em\null\nobreak\hfil\squareforqed
\parfillskip=0pt\finalhyphendemerits=0\endgraf}\fi}
\def\endenv{\ifmmode\;\else{\unskip\nobreak\hfil
\penalty50\hskip1em\null\nobreak\hfil\;
\parfillskip=0pt\finalhyphendemerits=0\endgraf}\fi}

\newcounter{remark}
\newenvironment{remark}[1][]{\refstepcounter{remark}\par\medskip\noindent%
\textbf{Remark~\theremark #1} }{\medskip}

\newcounter{example}

\mathchardef\ordinarycolon\mathcode`\:
\mathcode`\:=\string"8000
\def\vcentcolon{\mathrel{\mathop\ordinarycolon}}
\begingroup \catcode`\:=\active
  \lowercase{\endgroup
  \let :\vcentcolon
  }

\usepackage{colortbl}
\usepackage{pifont}
\definecolor{Gray}{gray}{0.92}
\definecolor{Gray2}{gray}{0.75}
\definecolor{maroon}{cmyk}{0,0.87,0.68,0.32}
\usepackage{booktabs}
\usepackage{makecell}
\usepackage{diagbox}
\usepackage{multirow}

\newcommand{\nc}{\newcommand}
\nc{\rnc}{\renewcommand}
\nc{\bra}[1]{\langle#1|}
\nc{\ket}[1]{|#1\rangle}
\nc{\<}{\langle}
\rnc{\>}{\rangle}
\nc{\ketbra}[2]{|#1\rangle\!\langle#2|}
\nc{\braket}[2]{\langle#1|#2\rangle}
\nc{\braandket}[3]{\langle #1|#2|#3\rangle}
\nc{\proj}[1]{| #1\rangle\!\langle #1 |}
\nc{\avg}[1]{\langle#1\rangle}
\nc{\Rank}{\operatorname{Rank}}
\nc{\smfrac}[2]{\mbox{$\frac{#1}{#2}$}}
\nc{\tr}{\operatorname{Tr}}
\nc{\ox}{\otimes}

\nc{\cA}{{\cal A}}
\nc{\cB}{{\cal B}}
\nc{\cC}{{\cal C}}
\nc{\cD}{{\cal D}}
\nc{\cE}{{\cal E}}
\nc{\cF}{{\cal F}}
\nc{\cG}{{\cal G}}
\nc{\cH}{{\cal H}}
\nc{\cI}{{\cal I}}
\nc{\cJ}{{\cal J}}
\nc{\cK}{{\cal K}}
\nc{\cL}{{\cal L}}
\nc{\cM}{{\cal M}}
\nc{\cN}{{\cal N}}
\nc{\cO}{{\cal O}}
\nc{\cP}{{\cal P}}
\nc{\cQ}{{\cal Q}}
\nc{\cR}{{\cal R}}
\nc{\cS}{{\cal S}}
\nc{\cT}{{\cal T}}
\nc{\cV}{{\cal V}}
\nc{\cU}{{\cal U}}
\nc{\cX}{{\cal X}}
\nc{\cY}{{\cal Y}}
\nc{\cZ}{{\cal Z}}
\nc{\cW}{{\cal W}}

\nc{\RR}{{{\mathbb R}}}
\nc{\CC}{{{\mathbb C}}}
\nc{\FF}{{{\mathbb F}}}
\nc{\NN}{{{\mathbb N}}}
\nc{\ZZ}{{{\mathbb Z}}}
\nc{\QQ}{{{\mathbb Q}}}
\nc{\UU}{{{\mathbb U}}}
\nc{\EE}{{{\mathbb E}}}
\nc{\id}{{\operatorname{id}}}

\nc{\1}{{\mathds{1}}}
\nc{\supp}{{\operatorname{supp}}}

\newcommand{\floor}[1]{\lfloor #1 \rfloor}

\def\ve{\varepsilon}

\makeatletter
\def\grd@save@target#1{%
  \def\grd@target{#1}}
\def\grd@save@start#1{%
  \def\grd@start{#1}}
\tikzset{
  grid with coordinates/.style={
    to path={%
      \pgfextra{%
        \edef\grd@@target{(\tikztotarget)}%
        \tikz@scan@one@point\grd@save@target\grd@@target\relax
        \edef\grd@@start{(\tikztostart)}%
        \tikz@scan@one@point\grd@save@start\grd@@start\relax
        \draw[minor help lines,magenta] (\tikztostart) grid (\tikztotarget);
        \draw[major help lines] (\tikztostart) grid (\tikztotarget);
        \grd@start
        \pgfmathsetmacro{\grd@xa}{\the\pgf@x/1cm}
        \pgfmathsetmacro{\grd@ya}{\the\pgf@y/1cm}
        \grd@target
        \pgfmathsetmacro{\grd@xb}{\the\pgf@x/1cm}
        \pgfmathsetmacro{\grd@yb}{\the\pgf@y/1cm}
        \pgfmathsetmacro{\grd@xc}{\grd@xa + \pgfkeysvalueof{/tikz/grid with coordinates/major step}}
        \pgfmathsetmacro{\grd@yc}{\grd@ya + \pgfkeysvalueof{/tikz/grid with coordinates/major step}}
        \foreach \x in {\grd@xa,\grd@xc,...,\grd@xb}
        \node[anchor=north] at (\x,\grd@ya) {\pgfmathprintnumber{\x}};
        \foreach \y in {\grd@ya,\grd@yc,...,\grd@yb}
        \node[anchor=east] at (\grd@xa,\y) {\pgfmathprintnumber{\y}};
      }
    }
  },
  minor help lines/.style={
    help lines,
    step=\pgfkeysvalueof{/tikz/grid with coordinates/minor step}
  },
  major help lines/.style={
    help lines,
    line width=\pgfkeysvalueof{/tikz/grid with coordinates/major line width},
    step=\pgfkeysvalueof{/tikz/grid with coordinates/major step}
  },
  grid with coordinates/.cd,
  minor step/.initial=.2,
  major step/.initial=1,
  major line width/.initial=2pt,
}
\makeatother

\makeatletter
\newcommand*\rel@kern[1]{\kern#1\dimexpr\macc@kerna}
\newcommand*\widebar[1]{%
  \begingroup
  \def\mathaccent##1##2{%
    \rel@kern{0.8}%
    \overline{\rel@kern{-0.8}\macc@nucleus\rel@kern{0.2}}%
    \rel@kern{-0.2}%
  }%
  \macc@depth\@ne
  \let\math@bgroup\@empty \let\math@egroup\macc@set@skewchar
  \mathsurround\z@ \frozen@everymath{\mathgroup\macc@group\relax}%
  \macc@set@skewchar\relax
  \let\mathaccentV\macc@nested@a
  \macc@nested@a\relax111{#1}%
  \endgroup
}
\makeatother

\newcommand{\reg}{\infty}

\newcommand{\Renyi}{R\'{e}nyi }

\usepackage{mathrsfs}
\nc{\pl}{{\scalebox{0.7}{+}}}
\nc{\PSD}{\HERM_{\pl}}
\nc{\PD}{\HERM_{\pl\pl}}
\nc{\polarPSD}[1]{{#1}_{\pl}^{\circ}}
\nc{\polarPSDre}[1]{{#1}_{\pl}^{\star}}
\nc{\polarPD}[1]{{#1}_{\pl\pl}^{\circ}}
\nc{\HERM}{\mathscr{H}}
\nc{\cvxset}{\mathscr{C}}
\nc{\density}{\mathscr{D}}
\nc{\subdensity}{\mathscr{D}_\bullet}
\nc{\bcC}{\cC}
\nc{\Meas}{{\scriptscriptstyle \rm M}}
\nc{\Proj}{{{\scriptscriptstyle \rm P}}}
\nc{\RM}{{{\mathscr{R}}}}
\nc{\sK}{{{\mathscr{K}}}}
\nc{\sS}{{{\mathscr{S}}}}
\nc{\sT}{{{\mathscr{T}}}}
\nc{\sA}{{{\mathscr{A}}}}
\nc{\sB}{{{\mathscr{B}}}}
\nc{\sC}{{{\mathscr{C}}}}
\nc{\sE}{{{\mathscr{E}}}}
\nc{\sL}{{{\mathscr{L}}}}
\nc{\sG}{{{\mathscr{G}}}}
\nc{\sF}{{{\mathscr{F}}}}
\nc{\sI}{{{\mathscr{I}}}}
\nc{\sN}{{{\mathscr{N}}}}
\nc{\sM}{{{\mathscr{M}}}}
\nc{\END}{\operatorname{End}}
\nc{\PERM}{\mathfrak{\sigma}}

\nc{\Cone}{\text{\rm Cone}}
\nc{\sep}{{\SEP}}
\nc{\BS}{{\scriptscriptstyle \rm {BS}}}
\nc{\Sand}{{\scriptscriptstyle  \rm S}}
\nc{\Petz}{{\scriptscriptstyle  \rm P}}
\nc{\Hypo}{{\scriptscriptstyle  \rm H}}
\nc{\DD}{{{\mathbb D}}}
\nc{\suchthat}{\text{\rm s.t.}}
\nc{\PPT}{\text{\rm PPT}}
\nc{\Rains}{\text{\rm Rains}}
\nc{\WD}{\text{\rm WD}}
\nc{\new}{\text{\rm new}}
\nc{\sfT}{\mathsf T}
\nc{\SEP}{\text{\rm SEP}}
\nc{\PSEP}{\text{\rm PSEP}}
\nc{\CPTP}{\text{\rm CPTP}}
\nc{\POVM}{\text{\rm POVM}}
\nc{\PVM}{\text{\rm PVM}}
\nc{\CP}{\text{\rm CP}}
\nc{\adv}{\text{\rm adv}}
\nc{\spec}{\text{\rm spec}}
\nc{\poly}{\text{\rm poly}}
\nc{\End}{\operatorname{End}}
\nc{\Par}{\operatorname{Par}}
\nc{\RNG}{\operatorname{RNG}}
\nc{\epi}{\boldsymbol{\operatorname{epi}}}
\nc{\op}{\boldsymbol{\operatorname{op}}}

\usepackage{stmaryrd}
\nc{\db}[1]{\left\llbracket#1 \right\rrbracket}

\usepackage{pgfplots}

\nc{\img}{\mathbf{i}}

\usepackage{changepage}

\begin{document}

\title{\Large \textbf{Generalized quantum asymptotic equipartition}}

\author[1]{Kun Fang \thanks{kunfang@cuhk.edu.cn}}
\author[2]{Hamza Fawzi \thanks{h.fawzi@damtp.cam.ac.uk}}
\author[3]{Omar Fawzi \thanks{omar.fawzi@ens-lyon.fr}}

\affil[1]{\small School of Data Science, The Chinese University of Hong Kong, Shenzhen,\protect\\  Guangdong, 518172, China}
\affil[2]{\small Department of Applied Mathematics and Theoretical Physics,  University of Cambridge, \protect\\ Cambridge CB3 0WA, United Kingdom}
\affil[3]{\small Inria, ENS de Lyon, UCBL, LIP, 69342, Lyon Cedex 07, France}

\date{}

\maketitle

\begin{abstract}
The asymptotic equipartition property (AEP) states that in the limit of a large number of independent and identically distributed (i.i.d.) random experiments, the output sequence is virtually certain to come from the typical set, each member of which is almost equally likely. This property is a form of the law of large numbers and lies at the heart of information theory. In this work, we prove a generalized quantum AEP beyond the i.i.d. framework where the random samples are drawn from two sets of quantum states. In particular, under suitable assumptions on the sets, we prove that all operationally relevant divergences converge to the quantum relative entropy between the sets. More specifically, both the quantum hypothesis testing relative entropy (a smoothed form of the min-relative entropy) and the smoothed max-relative entropy approach the regularized relative entropy between the sets. Notably, the asymptotic limit has explicit convergence guarantees and can be efficiently estimated through convex optimization programs, despite the regularization, provided that the sets have efficient descriptions. The generalized AEP directly implies a new {quantum Stein's lemma} for conducting quantum hypothesis testing between two sets of quantum states. 
{Combined with the accompanying work [IEEE TIT 72(4):2330-2342 (2026)],} this addresses open questions
raised by Brand\~{a}o et al.~[IEEE TIT 66(8):5037–5054 (2020)] and Mosonyi et al.~[IEEE TIT 68(2):1032-1067 (2022)], which seek a Stein’s lemma with computational efficiency. 
Moreover, we propose a new framework for quantum resource theory in which state transformations are performed without requiring precise characterization of the states being manipulated, making it more robust to imperfections. We demonstrate the reversibility (also referred to as the second law) of such a theory and identify the regularized relative entropy as the unique measure of the resource in this new framework. 

\end{abstract}

{
\hypersetup{linkcolor=black}
\tableofcontents
}

\section{Introduction}

\subsection{Background}

The asymptotic equipartition property (AEP) is a form of the law of large numbers in information theory, stating that $-\frac{1}{n} \log {p(X_1,X_2,\cdots, X_n)}$ is close to the entropy $H(X)$, where $X_i$ are independent, identically distributed (i.i.d.) random variables and $p(X_1,X_2,\cdots,X_n)$ is the probability of observing the sequence $X_1,X_2,\cdots,X_n$ \cite{cover1999elements}. The AEP was first stated by Shannon in his seminal 1948 paper~\cite{Shannon1948}, where he proved the result for i.i.d. processes and stated the result for {stationary} ergodic processes. McMillan~\cite{mcmillan1953basic} and Breiman~\cite{breiman1957individual} then proved the AEP for ergodic finite alphabet sources. The result is now referred to as the AEP or the Shannon–McMillan–Breiman theorem. This property lies {at} the heart of information theory, forming the mathematical foundation for applications such as data compression~\footnote{The AEP has been proved to be almost equivalent to the source coding theorem by~\cite{verdu1997role}, which reinforces the prominent role played by the AEP in information theory.}, channel coding~\cite{Shannon1948} and cryptography~\cite{han1993approximation,vembu1995generating,cachin1997entropy}.

The AEP can be represented in a more generic form in terms of entropic quantities:
\begin{align}
    \lim_{\ve\to 0} \lim_{n\to \infty} \frac{1}{n} \DD_{\ve}(P^{\ox n}\|Q^{\ox n}) = D(P\|Q),
\end{align}
where $P$ is a probability distribution, $Q$ is a nonnegative function on a finite set and $\DD$ is a divergence of interest, $\ve$ is a smoothing parameter and $D$ is the Kullback–Leibler divergence, also called the relative entropy~\cite{kullback1951information}. For instance, considering $Q$ to be a constant function equal to $1$ and $\DD$ be the min-relative entropy {or} max-relative entropy, then the above equation reduce to the {result by Shannon}.

The AEP has been generalized to quantum information theory, by replacing the probability distribution with the density matrix of a quantum state and using quantum divergences,
\begin{align}
    \lim_{\ve\to 0} \lim_{n\to \infty} \frac{1}{n} \DD_{\ve}(\rho^{\ox n}\|\sigma^{\ox n}) = D(\rho\|\sigma),
\end{align}
where $\rho$ is the density matrix, $\sigma$ is a positive semidefinite operator. In particular, Hiai and Petz~\cite{hiai1991proper} proved the case when $\DD_\ve$ is the quantum hypothesis testing relative entropy and $D$ is the Umegaki relative entropy. Ogawa and Nagaoka~\cite{Ogawa2000} then {strengthened} the result by removing the dependence {on} $\ve$ in the outer limit. These together formed the well-known quantum Stein's lemma in quantum information theory. Tomamichel et al.~\cite{tomamichel2009fully} proved the case when $\rho_{AB}$ is a bipartite quantum state and $\sigma_{AB} = I_A \ox \rho_B$ and $\DD_\ve$ is the quantum min-relative entropy {or} quantum max-relative entropy~\cite{datta2009min}. As the classical AEP plays an essential role in classical information theory, the quantum AEP has also found plenty of applications, including quantum data compression~\cite{schumacher1995quantum}, quantum state merging~\cite{berta2009single}, quantum channel coding~\cite{datta2013smooth}, quantum cryptography~\cite{renner2005security} and quantum resource theory~\cite{fang2019non}.

For many problems, the situation is not so simple and $\rho^{\ox n}$ and $\sigma^{\ox n}$ are not fully known and do not necessarily possess the i.i.d. structure. All one knows is that they belong to two sets of quantum states, $\sA_n$ and $\sB_n$, respectively, making the existing quantum AEPs not applicable. Such as scenario has been extensively studied in the classical information theory and applied to practical problems (refer to~\cite{levitan2002competitive} and references therein) such as classification with training sequences (e.g. speech recognition, signal detection and digital communication) and detection of messages via unknown channels (e.g. radar target detection, identification problem and watermark detection). This motivates us to explore a more general notion of the quantum AEP by extending beyond the i.i.d. structure of the source and beyond the singleton case. Specifically, we consider the following limit:
\begin{align}\label{eq: introduction quantum AEP}
    \lim_{n \to \infty} \frac{1}{n} \DD_{\ve}(\sA_n\|\sB_n) =\; ? \quad \text{with} \quad \DD_{\ve}(\sA_n\|\sB_n):= \inf_{\substack{\rho_n \in \sA_n\\ \sigma_n \in \sB_n}} \DD_{\ve}(\rho_n\|\sigma_n)
\end{align}
where $\sA_n$ is a set of quantum states and $\sB_n$ is a set of positive semidefinite operators acting on $\cH^{\otimes n}$, where $\cH$ is a finite-dimensional Hilbert space. This is a very general framework that encompasses almost all existing studies of quantum AEP in the literature. 
For instance, if $\DD_{\ve}$ is the hypothesis testing relative entropy, then $\sA_n$ is seen as a composite null hypothesis and $\sB_n$ as a composite alternate hypothesis. Specific examples that were recently studied include composite hypothesis testing where $\sA_n$ and $\sB_n$ are convex mixtures of i.i.d. states~\cite{berta2021composite},
as well as the generalized quantum Stein's lemma~\cite{Brand_o_2010,hayashi2024generalized,lami2024solutiongeneralisedquantumsteins}, where $\sA_n$ is an i.i.d. state {(or the set of ``almost power states''~\cite{lami2024solutiongeneralisedquantumsteins})} and $\sB_n$ is a set of quantum states.

\subsection{Main Results}

\subsubsection*{Generalized quantum AEP}

Under some structural assumptions on the sets $\sA_n$ and $\sB_n$, in particular the stability of the sets and their polars under tensor products, we prove that the limit in Eq.~\eqref{eq: introduction quantum AEP} exists for both the hypothesis testing relative entropy $D_{\Hypo,\ve}$ (a smoothed
form of the min-relative entropy) and the max-relative entropy $D_{\max,\ve}$, which represent two extreme cases in a family of quantum divergences. Both limits converge to the quantum relative entropy between the sets. More specifically, we establish the following:

\begin{theorem}(Generalized AEP, informal) \label{thm: generalized AEP intro}
Let $\{\sA_n\}_{n\in\NN}$ and $\{\sB_n\}_{n\in\NN}$ be two sequences of sets of states such that each set is convex, compact, permutation invariant, and the sequences as well as their polar sets are closed under tensor product.
Then for any $\ve \in (0,1)$, it holds that
\begin{alignat}{2}\label{eq: intro hyp testing rel entropy}
    n D^{\reg}(\sA \| \sB) & - O(n^{2/3} \log n) \leq D_{\Hypo,\ve}(\sA_n\|\sB_n) & \leq n D^{\reg}(\sA \| \sB) + O(n^{2/3}\log n),\\
    n D^{\reg}(\sA \| \sB) & - O(n^{2/3} \log n) \leq D_{\max,\ve}(\sA_n\|\sB_n) & \leq n D^{\reg}(\sA \| \sB) + O(n^{2/3}\log n),\label{eq: intro dmax}
\end{alignat}
where $D^{\reg}(\sA \| \sB) := \lim_{n\to \infty} \frac{1}{n}D(\sA_n\|\sB_n)$, which can be estimated with explicit convergence guarantees and can be efficiently computed given that $\sA_m$ and $\sB_m$ have efficient descriptions.
\end{theorem}

A central ingredient to the proof of this result is the superadditivity of the measured relative entropy between two sets of quantum states under the stability assumption of the polar set,
\begin{align}
    D_{\Meas}(\sA_{m+k}\|\sB_{m+k}) \geq D_{\Meas}(\sA_m\|\sB_m) + D_{\Meas}(\sA_k\|\sB_k). 
\end{align}
This superadditivity result makes crucial use of the variational expression of the measured relative entropy~\cite{petzquantum,Berta2017}. We note that~\cite{rubboli2024new,rubboli2024mixed} used a condition that is similar to the stability assumption of the polar set to show additivity of the relative entropy of entanglement for various classes of states. We now proceed with some remarks concerning Theorem~\ref{thm: generalized AEP intro}:
\begin{itemize}
    \item (\emph{Generality.}) All assumptions, except for the one regarding polar sets, are standard in most existing literature (e.g.,~\cite{Brand_o_2010,brandao2020adversarial,hayashi2024generalized,lami2024solutiongeneralisedquantumsteins}). The polar assumption, while technical, is satisfied in many cases of interest (see Table~\ref{tab: A list of sets that satisfy Assumption} later). For instance, it holds when the set is a singleton of i.i.d. state, as in the existing quantum AEP; when the set consists of the identity operator tensored with all density operators, which is relevant for conditional quantum entropies; or when it is the image of a quantum channel or a set of quantum channels, a scenario that naturally arises in quantum channel discrimination~\cite{fang2025adversarial}. 
    \item (\emph{Relaxation.}) Even in cases where the polar assumption is not directly satisfied, we can relax the original problem to meet the required conditions. This approach provides an efficient approximation to the original problem, which may be difficult to solve directly. We illustrate this approach in~\cite{fang2025efficient} through applications in entanglement manipulation and magic state distillation and anticipate that this approach has the potential for more applications beyond the specific cases discussed therein.
    \item (\emph{Efficiency.}) The leading order in the generalized AEP is characterized by a regularized quantity, which will be shown to be necessary in general through specific examples. However, the asymptotic limit comes with explicit convergence guarantees and can still be efficiently estimated via convex optimization programs, despite the regularization, provided that $\sA_n$ and $\sB_n$ have efficient descriptions. The detail of this computational aspect is provided in the accompanying paper~\cite{fang2025efficient}.
    \item (\emph{Explicit finite $n$ estimate.}) Our generalized AEP has a leading term that is independent of $\ve$, corresponding to the strong converse property in information theory. Moreover, while the leading order is characterized by a regularized quantity, we can still provide an explicit estimate for finite $n$, making its convergence controllable in the finite block length regime. Even though the bound we obtain on the second order term is $O(n^{2/3} \log n)$ instead of the usual $O(\sqrt{n})$, this is a rare case where we can explicitly bound the convergence rate for a regularized quantity. 
    \item (\emph{Divergences.}) The results apply to the two extreme cases of quantum divergence; therefore, any divergence that lies between the hypothesis testing relative entropy and the max-relative entropy, or is equivalent to these, will yield a similar result.
\end{itemize}

\subsubsection*{{Quantum Stein's lemma with composite hypotheses}}

Quantum hypothesis testing is a pivotal concept in quantum information theory, forming the foundation for numerous applications, including quantum channel coding, quantum illumination, and the operational interpretation of quantum divergences. A central task in hypothesis testing is to discriminate between quantum states corresponding to the different hypotheses, with the goal of determining which state is actually provided. In this work, we investigate quantum hypothesis testing between two sets of quantum states, $\sA_n$ and $\sB_n$. The fundamental difficulty in analyzing such a hypothesis testing is that both hypotheses are composite and extend beyond the i.i.d. structure. We need to control the type-I error $\alpha(\sA_n, M_n)$ within a constant threshold whenever the sample is drawn from $\sA_n$, in order to prevent the extremely bad case from passing the test. The asymmetric setting then seeks to determine the optimal exponent at which the type-II error probability $\beta(\sB_n, M_n)$ decays, known as the \emph{Stein's exponent}, while keeping the type-I error within a fixed threshold $\ve$. Specifically, the goal is to evaluate
\begin{align}
\beta_\ve(\sA_n\|\sB_n) := \inf_{0\leq M_n\leq I}\left\{\beta(\sB_n, M_n): \alpha(\sA_n, M_n) \leq \ve\right\}.
\end{align}

\begin{theorem}({Quantum Stein's lemma with composite hypotheses}, informal.)\label{thm: generalized steins intro}
Given the same assumptions on $\{\sA_n\}_{n\in\NN}$ and $\{\sB_n\}_{n\in\NN}$ as stated in Theorem~\ref{thm: generalized AEP intro}.
Then for any $\ve \in (0,1)$, 
\begin{align}\label{eq: adversarial quantum Steins lemma}
    \lim_{n\to \infty} - \frac{1}{n} \log \beta_\ve(\sA_n\|\sB_n) = D^\reg(\sA\|\sB).
\end{align}
\end{theorem}
This result extends the well-known quantum Stein's lemma to a significantly broader setting. Notably, by leveraging the generalized AEP in Theorem~\ref{thm: generalized AEP}, the rate of convergence in our version of Stein's lemma can be explicitly quantified. Furthermore, the Stein's exponent can be efficiently computed via convex optimization programs, despite the regularization involved, provided that $\sA_n$ and $\sB_n$ have efficient descriptions. {Combined with the accompanying work~\cite{fang2025efficient},} this addresses open questions of considerable interest in the community, particularly those seeking a Stein's lemma with efficient estimation, as discussed in~\cite[Section 3.5]{brandao2020adversarial} and~\cite[Section VI]{Mosonyi_2022}.

\vspace{0.2cm}
The discussion on the reduction to the classical case and the comparison with various versions of the quantum Stein's lemma is presented below:

\begin{itemize}
\item (\emph{Reduction to the classical case.}) If $\{\sA_n\}_{n\in\NN}$ and $\{\sB_n\}_{n\in\NN}$ are classical probability distributions, then the measured relative entropy coincides with the quantum relative entropy. In this case, we can show that $D(\sA_n\|\sB_n)$ is additive. Therefore, the Stein's exponent $D^\reg(\sA\|\sB)$ can be replaced by the {single-letter} quantity $D(\sA_1\|\sB_1)$ without regularization.
\item (\emph{Quantum Stein's lemma.}) Let $\sA_n = \{\rho^{\ox n}\}$ and $\sB_n = \{\sigma^{\ox n}\}$ be two singletons. Then it is easy to check that these sets satisfy all the required assumptions. Therefore, our {quantum Stein's lemma with composite hypotheses} covers the quantum Stein's lemma~\cite{hiai1991proper,Ogawa2000} as a special case.
\item (\emph{Generalized quantum Stein's lemma.}) 
Theorem~\ref{thm: generalized steins intro} is incomparable to the generalized quantum Stein lemma of~\cite{Brand_o_2010,hayashi2024generalized,lami2024solutiongeneralisedquantumsteins}. It is weaker in the sense that we have the additional assumption on the stability of the polar sets under tensor product for the alternate hypothesis $\sB_n$, but it is stronger at least in two ways: it allows for a composite null hypothesis $\sA_n$ whereas the results of~\cite{Brand_o_2010,hayashi2024generalized,lami2024solutiongeneralisedquantumsteins} are restricted to $\sA_n = \{\rho^{\otimes n}\}$ {(or the set of ``almost power states''~\cite{lami2024solutiongeneralisedquantumsteins})}, and in addition we can obtain efficient and controlled approximations of the Stein's exponent as presented in Theorem~\ref{thm: generalized AEP intro}.
\item (\emph{Composite quantum Stein's lemma.})  The composite quantum hypothesis testing studied in~\cite{berta2021composite,Mosonyi_2022} does not satisfy the stability assumption of the sets. Therefore, it is not comparable to our result or the previous generalized quantum Stein's lemma~\cite{Brand_o_2010,hayashi2024generalized,lami2024solutiongeneralisedquantumsteins}. 
\item (\emph{Quantum Stein's lemma for restricted measurements.}) In~\cite[Theorem 16]{brandao2020adversarial}, the authors proved a quantum Stein's lemma for restricted measurements, assuming that the measurements and the set of quantum states are compatible~\cite{Piani2009}. In this work, we do not impose constraints on the performed measurements. As a result, the measurements and the set of quantum states in our result may not be compatible in their context, making our result not directly comparable to theirs. It is also worth noting that our Stein's lemma possesses the strong converse property (i.e., its convergence is independent of $\ve$), whereas the Stein's lemma for restricted measurements~\cite[Theorem 16]{brandao2020adversarial} only {establishes} a weak converse property (i.e., it requires taking $\ve \to 0$).
\end{itemize}

\subsubsection*{Second law of quantum resource theory}

The existing generalized quantum Stein's lemma is closely related to the second law of quantum resource theory, which states that transformations between $\rho^{\ox n}$ and $\sigma^{\ox m}$ can be performed reversibly~\cite{Brand_o_2010,hayashi2024generalized,lami2024solutiongeneralisedquantumsteins}. Based on the new {quantum Stein's lemma with composite hypotheses}, we propose an extended framework that enables state transformations without requiring precise characterization of the states being manipulated. This approach enhances the framework's robustness to imperfections. Specifically, let $\sF_n$ denote the set of free states and $\RNG$ represent the set of asymptotically resource non-generating operations. If it is possible to transform an unknown source state $\rho_n \in \sA_n$ into a specific target state $\sigma_m \in \sB_m$ under $\RNG$, we define $m/n$ as an achievable rate. We show that the optimal achievable rate in the asymptotic limit, denoted as $r\left(\sA\xrightarrow[]{\RNG} \sB\right)$, is characterized by the following theorem.

\begin{theorem}(Second law of resource manipulation, informal.)
    Let $\{\sA_n\}_{n\in\NN}$, $\{\sB_n\}_{n\in\NN}$ and  $\{\sF_n\}_{n\in\NN}$ be threes sequences of sets satisfying the same assumptions in Theorem~\ref{thm: generalized AEP intro} where $\sA_n \subseteq \density(\cH^{\ox n})$ and $\sB_n \subseteq \density(\cH^{\ox n})$ are the sets of source states and the sets of target states, respectively, and $\sF_n \subseteq \density(\cH^{\ox n})$ are the sets of free states. Then it holds that
    \begin{align}
        r\left(\sA\xrightarrow[]{\RNG} \sB\right) = \frac{D^\reg(\sA\|\sF)}{D^\reg(\sB\|\sF)}.
    \end{align}
    As a consequence, \begin{align}
        r\left(\sA\xrightarrow[]{\RNG} \sB\right)r\left(\sB\xrightarrow[]{\RNG} \sA\right) = 1.
    \end{align}
\end{theorem}

This result demonstrates the reversibility of the proposed framework, indicating that transformations between $\sA_n$ and $\sB_m$ can be performed without any loss. Moreover, the transformation rate is entirely characterized by the regularized quantum relative entropy, establishing it as the ``unique'' measure of resource within this new framework.

\subsection{Organization}

The organization of the remaining content is as follows. In Section~\ref{sec: preliminaries}, we introduce necessary notation and preliminary results on convex optimization and quantum divergences that will be used throughout this work. In Section~\ref{sec: Variational formula and superadditivity}, we present the variational formula and superadditivity of the measured relative entropy between two sets of quantum states, which serve as the main technical tools for proving our main results. In Section~\ref{sec: Generalized asymptotic equipartition property}, we present our main result on the generalized AEP and its formal proof. In Section~\ref{sec: Application 1: another generalized quantum Stein's lemma}, we derive a {quantum Stein's lemma with composite hypotheses} by applying the generalized AEP. In Section~\ref{sec: Quantum resource theory with partial information and its reversibility}, we introduce a new framework of quantum resource theory and prove its reversibility. 
In Section~\ref{sec: conclusion}, we give a conclusion of this work with several open problems for further study.

\section{Preliminaries}
\label{sec: preliminaries}

\subsection{Notation}

In this section we set the notation and define several quantities that will be used throughout this work. Some frequently used notation are summarized in Table~\ref{tab: state version}. Note that we label different physical systems by capital Latin letters and use these labels as subscripts to guide the reader by indicating which system a mathematical object belongs to. We drop the subscripts when they are evident in the context of an expression (or if we are not talking about a specific system).

\setlength\extrarowheight{2pt}
\begin{table}[H]
\centering
\begin{tabular}{l|l}
\toprule[2pt]
Notation & Descriptions\\
\hline
$\cH_A$ & Hilbert space on system $A$\\
$|A|$ & Dimension of $\cH_A$\\
$\sL(A)$ & Linear operators on $\cH_A$\\
$\HERM(A)$ & Hermitian operators on $\cH_A$\\
$\HERM_{\pl}(A)$ & Positive semidefinite operators on $\cH_A$\\
$\HERM_{\pl\pl}(A)$ & Positive definite operators on $\cH_A$\\
$\density(A)$ & Density matrices on $\cH_A$\\
$\sA,\sB,\sC$ & Set of linear operators\\
$\cvxset^{\circ}$ & Polar set $\cvxset^\circ := \{X: \tr[XY] \leq 1,  \forall\, Y\in \cvxset\}$ of $\cvxset$\\
$\polarPSD{\cvxset}$ & Polar set restricted to positive semidefinite cone $\cvxset^\circ \cap \PSD$\\
$\polarPD{\cvxset}$ & Polar set restricted to positive definite operators $\cvxset^\circ \cap \PD$\\
$\cvxset^{\star}$ & Reverse polar set $\cvxset^\star := \{X: \tr[XY] \geq 1,  \forall\, Y\in \cvxset\}$ of $\cvxset$\\
$\cvxset^{\star}_{\pl}$ & Reverse polar set restricted to positive semidefinite cone $\cvxset^\star \cap \PSD$\\
$\cvxset^{\star}_{\pl\pl}$ & Reverse polar set restricted to positive definite operators $\cvxset^\star \cap \PD$\\
$\CPTP$ & Completely positive and trace preserving maps\\
$\CP$ & Completely positive maps\\
$\log(x)$ & Logarithm of $x$ in base two\\
\bottomrule[2pt]  
\end{tabular}
\caption{\small Overview of notational conventions.}
\label{tab: state version}
\end{table}

\subsection{Polar set and support function}

In the following, we introduce the definitions of the polar set and support function, along with some fundamental results that will be used in our discussions.

\newcommand{\revh}{\underline{h}}

\begin{definition}[Polar set and support function]
Let $\cvxset \subseteq \HERM$ be a convex set. Its polar set is defined by
\begin{align}
\cvxset^\circ:= \{X: \tr[XY] \leq 1,  \forall\, Y\in \cvxset\} = \{X : h_{\cvxset}(X) \leq 1\}
\end{align} 
where $h_{\cvxset}$ is the support function of $\cvxset$:
\begin{align}
    h_{\cvxset}(\omega) := \sup_{\sigma \in \cvxset} \tr[\omega \sigma].
\end{align}
Let $\polarPSD{\cvxset}:= \cvxset^\circ \cap \PSD$ and  $\polarPD{\cvxset}:= \cvxset^\circ \cap \PD$ be the intersections with positive semidefinite operators and positive definite operators, respectively.

Similarly, the reverse polar set of $\cvxset$ is defined as:
\begin{align}
    \cvxset^{\star} &:= \{X : \tr[XY] \geq 1, \forall\, Y \in \cvxset \} = \{X : \revh_{\cvxset}(X) \geq 1\}
\end{align}
where $\revh_{\cvxset}$ is the reverse support function of $\cvxset$:
\begin{align}
    \revh_{\cvxset}(\omega) := \inf_{\sigma \in \cvxset} \tr[\omega \sigma].
\end{align}
\end{definition}

The support and reverse support functions are defined as convex programs.
When $\cvxset$ is contained in the cone of positive semidefinite operators, we have the following dual representation.

\begin{lemma}\label{lem: support dual}
Let $\cvxset \subseteq \PSD$ be a convex set contained in the cone of positive semidefinite operators, and let $\omega \in \PSD$. Then $h_{\cvxset}(\omega) = \inf\{t > 0: \omega \in t \polarPSD{\cvxset}\}$, where we set the infimum to be $+\infty$ if $\{t > 0: \omega \in t \polarPSD{\cvxset}\} = \emptyset$,
and $\revh_{\cvxset}(\omega) = \sup\{t > 0 : \omega \in t \cvxset^{\star}_{\pl}\}$, where we set the supremum to be $0$ if $\{t > 0 : \omega \in t \cvxset^{\star}_{\pl}\} = \emptyset$.
\end{lemma}
\begin{proof}
We write the proof for $h_{\sC}$, for $\revh_{\cvxset}$ the proof is similar. We start from ``$\geq$'' direction. Let $s > h_{\cvxset}(\omega)\geq 0$. We have $h_{\cvxset}(\omega/s) = h_{\cvxset}(\omega)/s < 1$. So $\omega / s \in \polarPSD{\cvxset}$ and therefore $\omega \in s \polarPSD{\cvxset}$. This implies that $s$ is a feasible solution for $\inf\{t > 0: \omega \in t \polarPSD{\cvxset}\}$, and we get $s \geq \inf\{t > 0: \omega \in t \polarPSD{\cvxset}\}$. As this holds for any $s > h_{\cvxset}(\omega)$, we get $h_{\cvxset}(\omega) \geq \inf\{t > 0: \omega \in t \polarPSD{\cvxset}\}$. Now we prove the ``$\leq$'' direction. Let $z$ be any feasible solution to $\inf\{t > 0: \omega \in t \polarPSD{\cvxset}\}$. We have $\omega = z X$ with $X \in \polarPSD{\cvxset}$. Then $h_{\cvxset}(\omega) = z h_{\cvxset}(X) \leq z$ because $X \in \polarPSD{\cvxset}$ implies $h_{\cvxset}(X) \leq 1$. As this holds for any feasible solution $z$, we have $h_{\cvxset}(\omega) \leq \inf\{t > 0: \omega \in t \polarPSD{\cvxset}\}$.
\end{proof}

\begin{definition}(Stability.)\label{def: closed under tensor product}
    Let $\cH_1$ and $\cH_2$ be finite-dimensional Hilbert spaces. Consider three sets $\sA_1 \subseteq \PSD(\cH_1)$, $\sA_2 \subseteq \PSD(\cH_2)$, and $\sA_{12} \subseteq \PSD(\cH_1 \otimes \cH_2)$. 
    We say that $\{\sA_1,\sA_2,\sA_{12}\}$ is {closed under tensor product} if for any $X_1 \in \sA_1$ and $X_2 \in \sA_2$, we have $X_1 \otimes X_2 \in \sA_{12}$. In short, we write $\sA_1 \ox \sA_2 \subseteq \sA_{12}$.
\end{definition}

The following lemma provides an equivalent condition for determining if the polar sets of interest are closed under tensor product, which can be easier to validate for specific examples.

\begin{lemma}\label{lema: polar set and support function}
Let $\cH_1$ and $\cH_2$ be finite-dimensional Hilbert spaces. Consider three sets $\sA_1 \subseteq \PSD(\cH_1)$, $\sA_2 \subseteq \PSD(\cH_2)$, and $\sA_{12} \subseteq \PSD(\cH_1 \otimes \cH_2)$. Their polar sets are closed under tensor product if and only if their support functions are sub-multiplicative. That is,
\begin{align}\label{eq: polar set and support function tmp1}
\polarPSD{(\sA_{1})} \ox \polarPSD{(\sA_{2})} \subseteq \polarPSD{(\sA_{12})} \iff h_{\sA_{12}}(X_1 \ox X_2) \leq h_{\sA_1}(X_1) h_{\sA_2}(X_2),\; \forall X_i \in \PSD(\cH_i).
\end{align}
Similarly, the reverse polar sets are closed under tensor product if and only if the reverse support functions are super-multiplicative. That is,
\begin{align}\label{eq: reverse polar set and support function tmp1}
(\sA_{1})^{\star}_{\pl} \ox (\sA_{2})^{\star}_{\pl} \subseteq (\sA_{12})^{\star}_{\pl} \iff \revh_{\sA_{12}}(X_1 \ox X_2) \geq \revh_{\sA_1}(X_1) \revh_{\sA_2}(X_2),\; \forall X_i \in \PSD(\cH_i).
\end{align}
\end{lemma}
\begin{proof}
We focus on establishing Eq.~\eqref{eq: polar set and support function tmp1}, Eq.~\eqref{eq: reverse polar set and support function tmp1} is analogous. We start from ``$\Rightarrow$'' direction. Consider any $X_1 \in \PSD(\cH_1), X_2 \in \PSD(\cH_2)$. Let $t_1$ and $t_2$ be any feasible solutions of the dual programs in Lemma~\ref{lem: support dual} for $h_{\sA_1}(X_1)$ and $h_{\sA_2}(X_2)$, respectively. 
Then we have $X_1 \in t_1 \polarPSD{(\sA_{1})}$ and $X_2 \in t_2 \polarPSD{(\sA_{2})}$. 
This implies $X_1\ox X_2 \in t_1 t_2 \polarPSD{(\sA_{1})} \ox \polarPSD{(\sA_{2})} \subseteq t_1t_2 \polarPSD{(\sA_{12})}$ and therefore $t_1t_2$ is also a feasible solution for the dual program of $h_{\sA_{12}}(X_1 \ox X_2)$. So $h_{\sA_{12}}(X_1 \ox X_2) \leq t_1t_2$. As this holds for any feasible solutions $t_1, t_2$, we have $h_{\sA_{12}}(X_1 \ox X_2) \leq h_{\sA_1}(X_1) h_{\sA_2}(X_2)$. Now we prove the ``$\Leftarrow$'' direction. For any $Y_1 \in \polarPSD{(\sA_{1})}$ and $Y_2 \in \polarPSD{(\sA_{2})}$, we have $0 \leq h_{\sA_1}(Y_1) \leq 1$ and $0 \leq h_{\sA_2}(Y_2) \leq 1$. This implies that $h_{\sA_{12}}(Y_1 \ox Y_2) \leq 1$ and therefore, $Y_1 \ox Y_2 \in \polarPSD{(\sA_{12})}$. This shows $\polarPSD{(\sA_{1})} \ox \polarPSD{(\sA_{2})} \subseteq \polarPSD{(\sA_{12})}$.
\end{proof}

\subsection{Quantum divergences}

A functional $\DD: \density \times \PSD \to \RR$ is a quantum divergence if it satisfies the data-processing inequality $\DD(\cE(\rho)\|\cE(\sigma)) \leq \DD(\rho\|\sigma)$ for any CPTP map $\cE$ and $(\rho,\sigma) \in \density \times \PSD$. In the following, we will introduce several quantum divergences and their fundamental properties, which will be used throughout this work. Additionally, we will define quantum divergences between two sets of quantum states, which will be the main quantity of interest in this work.

\begin{definition}(Petz \Renyi divergence~\cite{petz1986quasi}.)
Let $\alpha \in (0,1) \cup (1,+\infty)$. For any $\rho\in \density$ and $\sigma \in \PSD$, the Petz \Renyi divergence is defined by
\begin{align}\label{eq: Petz}
    D_{\Petz,\alpha}(\rho\|\sigma) := \begin{cases}
    \frac{1}{\alpha-1}\log\tr\left[\rho^\alpha\sigma^{1-\alpha}\right], \quad & {\text{if}\ (\alpha < 1 \ \text{and} \ \rho \not\perp \sigma) \text{ or } (\supp(\rho) \subseteq \supp(\sigma)),}\\
    +\infty, \quad & \text{otherwise}.
    \end{cases}
\end{align}
\end{definition}

\begin{definition}(Sandwiched \Renyi divergence~\cite{muller2013quantum,wilde2014strong}.)
Let $\alpha \in (0,1) \cup (1,+\infty)$. For any $\rho\in \density$ and $\sigma \in \PSD$, the sandwiched \Renyi divergence is defined by
\begin{align}\label{eq: Sandwiched}
    D_{\Sand,\alpha}(\rho\|\sigma) := \begin{cases}
    \frac{1}{\alpha-1}\log\tr\left[\sigma^{\frac{1-\alpha}{2\alpha}}\rho\sigma^{\frac{1-\alpha}{2\alpha}}\right]^\alpha, \quad & {\text{if}\ (\alpha < 1 \ \text{and} \ \rho \not\perp \sigma) \text{ or } (\supp(\rho) \subseteq \supp(\sigma)),}\\
    +\infty, \quad & \text{otherwise}.
    \end{cases}
\end{align}
\end{definition}

\begin{definition}(Umegaki relative entropy~\cite{umegaki1954conditional}.)
For any $\rho\in \density$ and $\sigma \in \PSD$, the Umegaki relative entropy is defined by
\begin{align}\label{eq: Umegaki}
    D(\rho\|\sigma):= \tr[\rho(\log \rho - \log \sigma)],
\end{align}
if $\supp(\rho) \subseteq \supp(\sigma)$ and $+\infty$ otherwise.
\end{definition}

When $\alpha \to 1$, both $D_{\Petz,\alpha}$ and $D_{\Sand,\alpha}$ converge to the Umegaki relative entropy~\cite{muller2013quantum,wilde2014strong},
\begin{align}\label{eq: state Renyi continuous}
  \lim_{\alpha \to 1} D_{\Petz,\alpha}(\rho\|\sigma) = \lim_{\alpha \to 1} D_{\Sand,\alpha}(\rho\|\sigma) = D(\rho\|\sigma).
\end{align}
When $\alpha \to 0^+$, the Petz \Renyi divergence converges to the min-relative entropy or the Petz \Renyi divergence of order $0$,
\begin{align}
    \lim_{\alpha \to 0^+} D_{\Petz,\alpha}(\rho\|\sigma) = D_{\min}(\rho\|\sigma) = D_{\Petz, 0}(\rho \| \sigma) := -\log \tr [\Pi_\rho \sigma],
\end{align}
with $\Pi_\rho$ the projection on the support of $\rho$. When $\alpha \to \infty$, the sandwiched \Renyi divergence converges to the max-relative entropy~\cite{datta2009min},
\begin{align}\label{eq: definition of Dmax}
\lim_{\alpha \to \infty} D_{\Sand,\alpha}(\rho\|\sigma) = D_{\max}(\rho\|\sigma):= \log\inf\big\{t \in \RR \;:\; \rho \leq t\sigma \big\}\;,
\end{align}
if $\supp(\rho) \subseteq \supp(\sigma)$ and $+\infty$ otherwise. Let $F(\rho,\sigma) := \|\sqrt{\rho}\sqrt{\sigma}\|_1+\sqrt{(1-\tr\rho)(1-\tr \sigma)}$ be the generalized fidelity and $P(\rho,\sigma):= \sqrt{1-F^2(\rho,\sigma)}$ be the purified distance. Let $\ve \in (0,1)$. Then the smoothed max-relative entropy is defined by
\begin{align}
  D_{\max,\ve}(\rho\|\sigma):= \inf_{\rho':P(\rho',\rho) \leq \ve} D_{\max}(\rho'\|\sigma),
\end{align}
where the infimum is taken over all subnormalized states that are $\ve$-close to the state $\rho$.

\begin{definition}(Hypothesis testing relative entropy.)
Let $\ve \in [0,1]$. For any $\rho\in \density$ and $\sigma \in \PSD$, the quantum hypothesis testing relative entropy is defined by $D_{\Hypo, \ve}(\rho\|\sigma) := -\log \beta_{\ve}(\rho\|\sigma)$
where
\begin{align}
    \beta_\ve(\rho\|\sigma): = \min_{0\leq M \leq I} \left\{\tr[\sigma M]: \tr[\rho(I-M)] \leq \ve\right\}.
\end{align}
\end{definition}

The following lemma combines~\cite[Lemma 5]{cooney2016strong} and~\cite[Proposition 3]{qi2018applications}.
\begin{lemma}(One-shot {relations}.)\label{lem: DH petz sandwiched}
Let $\alpha \in (0,1)$, $\alpha' \in (1,+\infty)$ and $\ve\in(0,1)$. For any $\rho\in \density$ and $\sigma \in \PSD$, it holds that
\begin{align}
D_{\Petz,\alpha}(\rho\|\sigma) + \frac{\alpha}{\alpha-1} \log \frac{1}{\ve} \leq D_{\Hypo,\ve}(\rho\|\sigma)\leq D_{\Sand,\alpha'}(\rho\|\sigma) + \frac{\alpha'}{\alpha'-1}\log\frac{1}{1-\ve}.
\end{align}
\end{lemma}

\begin{definition}(Measured relative entropy~\cite{donald1986relative,hiai1991proper}.)
For any $\rho \in \density$ and $\sigma \in \PSD$, the measured relative entropy is defined by
\begin{align}
D_{\Meas} (\rho\|\sigma) := \sup_{(\cX,M)} D(P_{\rho,M}\|P_{\sigma,M}),
\end{align}
where $D$ is the Kullback–Leibler divergence and the optimization is over finite sets $\cX$ and positive operator valued measures $M$ on $\cX$ such that $M_x \geq 0$ and $\sum_{x \in \cX} M_x = I$, $P_{\rho,M}$ is a measure on $\cX$ defined via the relation $P_{\rho,M}(x) = \tr[M_x\rho]$ for any $x \in \cX$.    
\end{definition}

A variational expression for $D_{\Meas}$ is given by~\cite[Lemma 1]{Berta2017},
\begin{align}\label{eq: DM variational}
    D_{\Meas}(\rho\|\sigma) = \sup_{\omega \in \PD} \ \tr[\rho \log \omega] + 1  - \tr[\sigma \omega].
\end{align}

\begin{definition}(Measured \Renyi divergence~\cite{Berta2017}.)
Let $\alpha \in (0,1) \cup (1,\infty)$.  For any $\rho \in \density$ and $\sigma \in \PSD$, the {measured \Renyi divergence} is defined as
\begin{align}\label{eq: definition DM alpha}
D_{\Meas, \alpha} (\rho\|\sigma) := \sup_{(\cX,M)} D_{\alpha}(P_{\rho,M}\|P_{\sigma,M}),
\end{align}
where $D_{\alpha}$ is the classical \Renyi divergence. 
\end{definition}
A variational expression for $D_{\Meas, \alpha}$ is given by~\cite[Lemma 3]{Berta2017},
\begin{align}\label{eq: DM alpha variational}
D_{\Meas,\alpha}(\rho\|\sigma) = \frac{1}{\alpha - 1} \log  \begin{cases}  \displaystyle \inf_{W \in \PD} \ \alpha \tr[\rho W] + (1-\alpha) \tr\left[\sigma W^{\frac{\alpha}{\alpha-1}}\right], & \text{if}\ \alpha \in (0,1/2),\\
\displaystyle \inf_{W \in \PD} \ \alpha \tr\left[\rho W^{\frac{\alpha-1}{\alpha}}\right] + (1-\alpha) \tr[\sigma W], & \text{if}\ \alpha \in [1/2,1),\\
\displaystyle \sup_{W \in \PD} \ \alpha \tr\left[\rho W^{\frac{\alpha-1}{\alpha}}\right] + (1-\alpha) \tr[\sigma W], & \text{if}\ \alpha \in (1,\infty).
\end{cases}
\end{align}
When $\alpha \to 1$, the measured \Renyi divergence converges to the measured relative entropy. It is known that~(e.g. \cite[Lemma 1]{rippchen2024locally})
$D_{\Meas,\alpha}$ is jointly convex for $\alpha \in (0,1]$ and jointly quasi-convex for $\alpha \in (1,\infty)$. For any $\cE \in \CPTP$ and $\alpha > 0$, the data-processing inequality holds (e.g.~\cite[Lemma 4]{rippchen2024locally})
\begin{align}
    D_{\Meas,\alpha}(\cE(\rho)\|\cE(\sigma)) \leq D_{\Meas,\alpha}(\rho\|\sigma).
\end{align}

{The following result establishes the ordering relations among various quantum relative entropies. While these relations are standard in quantum information theory, we include proofs here for completeness.}

\begin{lemma}\label{thm: comparison of quantum divergence}
Let $\alpha \in [1/2,1)$. For any $\rho \in \density$ and $\sigma \in \PSD$, 
\begin{align}
	D_{\min}(\rho\|\sigma) \leq D_{\Meas, \alpha}(\rho\|\sigma) \leq D_{\Meas}(\rho\|\sigma) \leq  D(\rho\|\sigma).
\end{align}
\end{lemma}
\begin{proof}
The last two inequalities follow from the monotonicity in $\alpha$ of the classical R\'enyi divergences and the data processing inequality for $D$. As $D_{\Meas, \alpha}$ is monotone increasing in $\alpha$, it remains to show the first inequality for $\alpha = \frac12$. By the variational formula in~\cite[Eq. (21)]{Berta2017}, we have
\begin{align}
D_{\Meas, 1/2}(\rho\|\sigma) = -\log \inf_{\omega \in \PD} \tr[\rho \omega^{-1}] \tr[\sigma \omega],
\end{align}
which is the same as the Alberti's theorem for quantum fidelity~(see e.g.~\cite[Corollary 3.20]{watrous2018theory}). Consider a feasible solution $\omega_\ve = \Pi_\rho + \ve(I - \Pi_\rho) \in \PD$ with $\ve > 0$. It gives $D_{\Meas, 1/2}(\rho\|\sigma) \geq -\log \tr[\rho \omega_\ve^{-1}] \tr[\sigma \omega_\ve]$. Since $\rho$ has trace one, it gives $\tr[\rho \omega_\ve^{-1}] = \tr[\rho] = 1$. Then we have \begin{align}
D_{\Meas, 1/2}(\rho\|\sigma) \geq -\log \tr[\sigma \omega_\ve] = -\log [(1-\ve) \tr \Pi_\rho \sigma + \ve].
\end{align} 
As the above holds for any $\ve > 0$, we take $\ve \to 0^{+}$ and get $D_{\Meas, 1/2}(\rho\|\sigma) \geq -\log \tr[\Pi_\rho \sigma] = D_{\min}(\rho\|\sigma)$, which completes the proof. 
{A simpler alternative proof for $D_{\Meas,1/2}(\rho\|\sigma) \geq D_{\min}(\rho\|\sigma)$ is to consider the POVM $\{\Pi_\rho, I - \Pi_\rho\}$ in the definition of the measured \Renyi divergence. This choice yields the achievable value $D_{\min}(\rho\|\sigma)$ by direct calculation.}
\end{proof}

\begin{lemma}\label{lem: DM and Sandwiched relation}
Let $\alpha \in [1/2,\infty)$. For any $\rho \in \density$ and $\sigma \in \PSD$, it hold
\begin{align}
    D_{\Meas,\alpha}(\rho\|\sigma) \leq D_{\Sand,\alpha}(\rho\|\sigma) \leq D_{\Meas,\alpha}(\rho\|\sigma) + 2\log |\spec(\sigma)|,
\end{align}
where $|\spec(\sigma)|$ is the number of mutually different eigenvalues of $\sigma$.
\end{lemma}
\begin{proof}
The first inequality is given by~\cite[Theorem 6]{Berta2017}. Now we prove the second inequality. Let $\cP_\sigma$ be a pinching map with respect to $\sigma$. Then we have 
\begin{align}
    D_{\Sand,\alpha}(\rho\|\sigma) & \leq D_{\Sand,\alpha}(\cP_\sigma(\rho)\|\sigma) + 2\log |\spec(\sigma)|\\
    & = D_{\Meas,\alpha}(\cP_\sigma(\rho)\|\sigma) + 2\log |\spec(\sigma)|\\
    & = D_{\Meas,\alpha}(\cP_\sigma(\rho)\|\cP_\sigma(\sigma)) + 2\log |\spec(\sigma)|\\
    & \leq D_{\Meas,\alpha}(\rho\|\sigma) + 2\log |\spec(\sigma)|
\end{align}
where the first line follows from~\cite[Lemma 3]{hayashi2016correlation}, the second line follows because $\cP_\sigma(\rho)$ and $\sigma$ commute and therefore $D_{\Sand,\alpha}(\cP_\sigma(\rho)\|\sigma) = D_{\Meas,\alpha}(\cP_\sigma(\rho)\|\sigma)$, the third line follows by $\cP_\sigma(\sigma) = \sigma$ and the last line follows from the data processing inequality.
\end{proof}

{The following is a standard result from representation theory and is widely used in quantum information theory; see, e.g.,~\cite[Section~5]{harrow2005applications},~\cite[Lemma 2.4]{berta2021composite} and~\cite[Lemma~A.1]{fawzi2021defining}.}

\begin{lemma}\label{lem: permutation invariant spec}
Let $X$ be a permutation-invarant operator on $\sL(A^{\ox n})$ with $|A| = d$. Then $|\spec(X)| \leq (n+1)^{d} (n+d)^{d^2} = \poly(n)$.
\end{lemma}

In this work, we will focus on the study of quantum divergences between two sets of quantum states.

\begin{definition}(Quantum divergence between two sets of states.)\label{def: divergence between two sets}
    Let $\DD$ be a quantum divergence between two quantum states. Then for any sets $\sA\subseteq \density$ and $\sB\subseteq \PSD$, the quantum divergence between these two sets of quantum states is defined by
    \begin{align}
        \DD(\sA\|\sB):= \inf_{\substack{\rho \in \sA\\ \sigma \in \sB}} \DD(\rho\|\sigma). 
    \end{align}
\end{definition}

Note that if $\DD$ is lower semicontinuous (which is true for most quantum divergences of interest {and even quantum relative entropy in infinite dimensions~\cite{lami2023attainability}}), and $\sA$ and $\sB$ are compact sets, the infimum in the above expression is always attained and can thus be replaced by a minimization~\cite[Theorem 7.3.1]{kurdila2005convex}.
From a geometric perspective, this quantity characterizes the distance between two sets $\sA$ and $\sB$ under the ``distance metric'' $\DD$. In particular, if $\sA = \{\rho\}$ is a singleton, we write $\DD(\rho\|\sB):= \DD(\{\rho\}\|\sB)$. For two sequences of sets $\{\sA_n\}_{n\in \NN}$ and $\{\sB_n\}_{n \in \NN}$, the regularized divergence is defined by 
\begin{align}
\DD^{\reg}(\sA \| \sB) := \lim_{n \to \infty} \frac{1}{n} \DD(\sA_{n} \| \sB_{n}),
\end{align}
whenever the limit on the right-hand side exists.

\section{Variational formula and superadditivity}
\label{sec: Variational formula and superadditivity}

Unlike other quantum relative entropies, which are defined by closed-form expressions, the measured relative entropy requires maximization over all possible quantum measurements, making it inherently challenging to compute. This complexity is compounded in the case of divergences between different sets of quantum states, where additional layers of optimization transform it into a minimax problem. The following results address this issue by providing variational formulas for $D_{\Meas}(\sA\|\sB)$ and $D_{\Meas,\alpha}(\sA\|\sB)$, reformulating them as more tractable convex optimization programs. Based on these variational formulas, we establish their superadditivity, which serves as a key technical tool for the main results of this work in the subsequent sections.

\begin{shaded}
\begin{lemma}\label{lem: DM variational sets}
Let $\sA\subseteq \density$ and $\sB \subseteq \PSD$ be two compact convex sets. Then it holds that
\begin{align}\label{eq: DM general convex program}
D_{\Meas}(\sA\|\sB) = \sup_{W \in \polarPD{\sB}} -h_{\sA}(-\log W),
\end{align}
where the objective function on the right-hand side is concave in $W$. 
\end{lemma}
\end{shaded}
\begin{proof}
By the variational formula of the measured relative entropy in Eq.~\eqref{eq: DM variational}, we have
\begin{align}
D_{\Meas}(\sA \| \sB) &= \inf_{\substack{\rho \in \sA\\ \sigma \in \sB}} \sup_{W \in \PD} \tr[\rho \log W] + 1 - \tr[\sigma W].
\end{align}
Note that all $\sA,\sB,\PD$ are convex sets, with $\sA,\sB$ being compact. Moreover, the objective function is linear in $(\rho,\sigma)$ and concave in $W$. So we can apply Sion's minimax theorem~\cite[Corollary 3.3]{Sion1958} to exchange the infimum and supremum and get
\begin{align}
D_{\Meas}(\sA \| \sB) & = \sup_{W \in \PD} \inf_{\substack{\rho \in \sA\\ \sigma \in \sB}} \tr[\rho \log W] + 1 - \tr[\sigma W]\\
&= \sup_{W \in \PD} \inf_{\rho \in \sA} \tr[\rho \log W] + 1 - \sup_{\sigma \in \sB} \tr[\sigma W].
\end{align}
By using the definition of the support function, we have
\begin{align}
    D_{\Meas}(\sA \| \sB)  &= \sup_{W \in \PD} -h_{\sA}(-\log W) + 1 - h_{\sB}(W).
\end{align}
Next, we aim to simplify the objective function by moving the term $1 - h_{\sB}(W)$ into the constraint. For this, let $W$ be a feasible solution in the above optimization and define $\widetilde W = W / h_{\sB}(W)$. Then $h_{\sB}(\widetilde W) = 1$ and we will see that $\widetilde W$ achieves an objective value no smaller than $W$. Indeed
\begin{align}
-h_{\sA}(-\log \widetilde{W})
& =  -h_{\sA}(-\log {W}) - \log h_{\sB}(W) \geq  -h_{\sA}(-\log {W}) + 1 - h_{\sB}(W)
\end{align}
where the inequality follows from the fact that $\log x \leq x-1$. 
Therefore, we can reformulate 
\begin{align}
D_{\Meas}(\sA \| \sB)  &= \sup_{\substack{W \in \PD\\h_{\sB}(W) = 1}} -h_{\sA}(-\log W).
\end{align}
Using the same argument as above, we relax the condition $h_{\sB}(W) = 1$ to $h_{\sB}(W) \leq 1$ and get 
\begin{align}
D_{\Meas}(\sA \| \sB) &= \sup_{\substack{W \in \PD\\h_{\sB}(W) \leq 1}} -h_{\sA}(-\log W).
\end{align}
Finally, noting that $h_{\sB}(W) \leq 1$ if and only if $W \in \sB^{\circ}$, we have the expression in Eq.~\eqref{eq: DM general convex program}. 
{Note that $-h_{\sA}(-\log W) = \inf_{\rho \in \sA} \tr[\rho \log W]$ is the pointwise infimum of concave functions (since the matrix logarithm is operator concave), and therefore $-h_{\sA}(-\log W)$ is itself concave.}
\end{proof}

Following the same argument, we can show the variational formula for the measured \Renyi divergence between two sets of quantum states. 

\begin{shaded}
\begin{lemma}\label{lem: DM alpha variational sets}
Let $\sA\subseteq \density$ and $\sB \subseteq \PSD$ be two compact convex sets. Then it holds that
\begin{align}
D_{\Meas, \alpha}(\sA\|\sB) = \frac{\alpha}{\alpha-1} \log (*)
\end{align}
with
\begin{align}
\label{eq: DM alpha variational sets}
(*) =
\begin{dcases}
\inf_{W, V} h_{\sA}(W)  & \text{\rm s.t.}\ W \in \PD, \; \; V \in \sB^\circ,\; W^{\frac{\alpha}{\alpha-1}} \leq V,\; \ \ \text{\rm if} \ \alpha \in (0,1/2),\\
\inf_{W, V} h_{\sA}(V)  & \text{\rm s.t.}\ W \in \PD, \; W \in \sB^{\circ},\; W^{\frac{\alpha - 1}{\alpha}} \leq V,\; \ \  \text{\rm if} \ \alpha \in [1/2,1),\\
\sup_{W, V} \revh_{\sA}(V) & \text{\rm s.t.}\ W \in \PD,\; W \in \sB^{\star},\; W^{\frac{\alpha - 1}{\alpha}} \geq V \geq 0,\; \ \ \text{\rm if} \ \alpha \in (1,+\infty),\\
\end{dcases}
\end{align}
where the right-hand side is a convex program.
\end{lemma}
\end{shaded}
\begin{proof}
We prove the result for $\alpha \in [1/2,1)$ here. The other parameter ranges are similar, and they are detailed in Appendix~\ref{app sec: variational formula for other alpha} for completeness.
For any fixed $\rho \in \sA$ and $\sigma \in \sB$, we have by Eq.~\eqref{eq: DM alpha variational} that
\begin{align}
Q_{\Meas,\alpha}(\rho\|\sigma) & := \inf_{W \in \PD}\ \alpha \tr \left[\rho W^{1-\frac{1}{\alpha}}\right] + (1-\alpha) \tr [\sigma W]\\
& = \inf_{\substack{W \in \PD\\W^{1-\frac{1}{\alpha}} \leq V}}\ \alpha \tr \left[\rho V\right] + (1-\alpha) \tr [\sigma W],
\end{align}
where the second equality follows by introducing an additional variable $V$.
Then we have
\begin{align}
\sup_{\substack{\rho \in \sA\\ \sigma \in \sB}} Q_\alpha^\Meas(\rho\|\sigma)
& = \sup_{\substack{\rho \in \sA\\ \sigma \in \sB}}  \; \inf_{\substack{W \in \PD\\W^{1-\frac{1}{\alpha}} \leq V}}\ \alpha \tr \left[\rho V\right] + (1-\alpha) \tr [\sigma W].
\end{align}
{Due to the operator convexity of the function $A \mapsto A^{r}$, for $-1 \le r < 0$,} all $\sA,\sB$ and $\{(W,V): W \in \PD, W^{1-\frac{1}{\alpha}} \leq V\}$ are convex sets, with $\sA,\sB$ being compact. Moreover, the objective function is linear in $(\rho,\sigma)$, and also linear in $(W,V)$. So we can apply Sion's minimax theorem~\cite[Corollary 3.3]{Sion1958} to exchange the infimum and supremum and get
\begin{align}
\sup_{\substack{\rho \in \sA\\ \sigma \in \sB}} Q_\alpha^\Meas(\rho\|\sigma) & {=} \inf_{\substack{W \in \PD\\W^{1-\frac{1}{\alpha}} \leq V}} \; \sup_{\substack{\rho \in \sA\\ \sigma \in \sB}} \ \alpha \tr \left[\rho V\right] + (1-\alpha) \tr [\sigma W]\\
& {=} \inf_{\substack{W \in \PD\\W^{1-\frac{1}{\alpha}} \leq V}} \;  \ \alpha \sup_{\rho \in \sA} \tr \left[\rho V\right] + (1-\alpha) \sup_{\sigma \in \sB} \tr [\sigma W].
\end{align}
By using the definition of the support function, we have
\begin{align}
\sup_{\substack{\rho \in \sA\\ \sigma \in \sB}} Q_\alpha^\Meas(\rho\|\sigma) & = \inf_{\substack{W \in \PD\\W^{1-\frac{1}{\alpha}} \leq V}} \alpha h_{\sA}(V) + (1-\alpha) h_{\sB}(W)\\
& = \inf_{\substack{W \in \PD\\W^{1-\frac{1}{\alpha}} \leq V}} h_{\sA}(V)^\alpha  h_{\sB}(W)^{1-\alpha},
\end{align}
where the second line follows from the weighed arithmetic-geometric mean inequality $\alpha x + (1-\alpha)y \leq x^\alpha y^{1-\alpha}$ (with equality if and only if $x = y$) and the fact that $(W,V)$ is a feasible solution implies $(kW, k^{1-1/\alpha} V)$ is also a feasible solution for any $k > 0$. Therefore, we can choose $k = (h_{\sA}(V))^\alpha (h_{\sB}(W))^{-\alpha}$, which implies $h_{\sA}(k^{1-1/\alpha}V) = h_{\sB}(kW)$ and therefore the equality of the weighed arithmetic-geometric mean is achieved. Similarly, for any feasible solution $(W, V)$ we can always construct a new solution $(W/h_{\sB}(W),\allowbreak Vh_{\sB}(W)^{1/\alpha - 1})$ achieves the same objective value. This implies 
\begin{align}
    \sup_{\substack{\rho \in \sA\\ \sigma \in \sB}} Q_\alpha^\Meas(\rho\|\sigma)  = \inf_{\substack{W \in \PD\\W^{1-\frac{1}{\alpha}} \leq V\\h_{\sB}(W) = 1}} (h_{\sA}(V))^\alpha = \inf_{\substack{W \in \PD\\W^{1-\frac{1}{\alpha}} \leq V\\h_{\sB}(W) \leq 1}} (h_{\sA}(V))^\alpha,
\end{align}
where the second equality follows by the same reasoning. 
Finally, noting that $h_{\sB}(W) \leq 1$ if and only if $W \in \sB^{\circ}$, we have the asserted result in Eq.~\eqref{eq: DM alpha variational sets}. It is easy to check that the objective function $h_{\sA}(V)$ is convex in $V$ and the feasible set is also a convex set.
\end{proof}

The variational formula helps establish the following superadditivity of the divergence between two sets of quantum states, which serves as a key technical tool for the applications discussed in the subsequent sections.

\begin{shaded}
    \begin{lemma}\label{lem: generalized supadditivity DM alpha}
    Let $\cH_{1}, \cH_{2}$ be finite-dimensional Hilbert spaces. Consider convex and compact sets $\sA_1 \subseteq \density(\cH_1)$, $\sA_2 \subseteq \density(\cH_2)$, and $\sA_{12} \subseteq \density(\cH_1 \otimes \cH_2)$ and also $\sB_1 \subseteq \PSD(\cH_1)$, $\sB_{2} \subseteq \PSD(\cH_2)$, and $\sB_{12} \subseteq \PSD(\cH_{1} \otimes \cH_{2})$. For  $\alpha \in (0,1)$, assume that $\polarPSD{(\sA_{1})} \ox \polarPSD{(\sA_{2})} \subseteq \polarPSD{(\sA_{12})}$ and $\polarPSD{(\sB_1)} \ox \polarPSD{(\sB_2)} \subseteq \polarPSD{(\sB_{12})}$  then,
    \begin{align}\label{eq: generalized supadditivity DM alpha}
        D_{\Meas,\alpha}(\sA_{12}\|\sB_{12}) \geq D_{\Meas,\alpha}(\sA_{1}\|\sB_{1}) + D_{\Meas,\alpha}(\sA_{2}\|\sB_{2}). 
    \end{align}
For $\alpha > 1$, Eq.~\eqref{eq: generalized supadditivity DM alpha} holds provided $\polarPSDre{(\sA_{1})} \ox \polarPSDre{(\sA_{2})} \subseteq \polarPSDre{(\sA_{12})}$ and $\polarPSDre{(\sB_1)} \ox \polarPSDre{(\sB_2)} \subseteq \polarPSDre{(\sB_{12})}$.
    \end{lemma}
\end{shaded}

\begin{proof}
Let $\alpha \in [1/2,1)$. Consider the convex program in Lemma~\ref{lem: DM alpha variational sets} for $D_{\Meas,\alpha}(\sA_1\|\sB_1)$ and  $D_{\Meas,\alpha}(\sA_2\|\sB_2)$. Let $(W_1, V_1)$ and $(W_2, V_2)$ be any feasible solutions for these programs, respectively. Then we will see that $(W_1\ox W_2, V_1\ox V_2)$ is also a feasible solution for $D_{\Meas,\alpha}(\sA_{12}\|\sB_{12})$ with a higher objective value. To see this, using the assumption, we have that $W_1 \in (\sB_1)_{\pl}^\circ$ and $W_2 \in (\sB_2)_\pl^\circ$ implies that $W_1 \ox W_2 \in (\sB_{12})_\pl^\circ$. Then by the multiplicativity of the power function, we have 
\begin{align}
    (W_1\ox W_2)^{\frac{\alpha-1}{\alpha}} = (W_1)^{\frac{\alpha-1}{\alpha}} \ox (W_2)^{\frac{\alpha-1}{\alpha}} \leq V_1 \ox V_2.
\end{align}
This confirms that $(W_1\ox W_2, V_1\ox V_2)$ is a feasible solution for the optimization of $D_{\Meas,\alpha}(\sA_{12}\|\sB_{12})$. By the assumption of $\polarPSD{(\sA_{1})} \ox \polarPSD{(\sA_{2})} \subseteq \polarPSD{(\sA_{12})}$ and Lemma~\ref{lema: polar set and support function}, we have $ h_{\sA_{12}}(V_1\ox V_2) \leq h_{\sA_1}(V_1) h_{\sA_2}(V_2)$, which implies
\begin{align}
    \frac{\alpha}{\alpha - 1} \log h_{\sA_{12}}(V_1\ox V_2) \geq \frac{\alpha}{\alpha - 1} \log h_{\sA_1}(V_1) + \frac{\alpha}{\alpha - 1} \log h_{\sA_2}(V_2).
\end{align} 
As the above relation holds for any feasible solutions, we have the asserted result in Eq.~\eqref{eq: generalized supadditivity DM alpha}. The other ranges for $\alpha$ follow the same reasoning.
\end{proof}

\begin{shaded}
\begin{lemma}\label{lem: general DM alpha continuitiy}
Let $\sA \subseteq \density$ and $\sB \subseteq \PSD$ be two compact convex sets. Then
\begin{align}\label{eq: general continuitiy measured channel worst case}
    \lim_{\alpha \to 1^-} D_{\Meas,\alpha}(\sA\|\sB) = \sup_{\alpha \in (0,1)} D_{\Meas,\alpha}(\sA\|\sB) = D_{\Meas}(\sA\|\sB).
\end{align}
\end{lemma}
\end{shaded}
\begin{proof}
Since $D_{\Meas,\alpha}$ is monotonically increasing in $\alpha$, we have for any $\alpha > \beta$, 
\begin{align}
D_{\Meas,\alpha}(\sA\|\sB) = \inf_{\substack{\rho\in\sA\\\sigma\in\sB}} D_{\Meas,\alpha}(\rho\|\sigma) \geq \inf_{\substack{\rho\in\sA\\\sigma\in\sB}} D_{\Meas,\beta}(\rho\|\sigma) = D_{\Meas,\beta}(\sA\|\sB).
\end{align}
So the divergence $D_{\Meas,\alpha}(\sA\|\sB)$ is also monotonically increasing in $\alpha$. This implies the first equality in Eq.~\eqref{eq: general continuitiy measured channel worst case}. Now we prove the second equality. Note that
$D_{\Meas,\alpha}(\rho\|\sigma)$ is lower semi-continous on $\PSD\times \PSD$ for every $\alpha > 0$~\cite[Proposition III.11]{mosonyi2023continuitypropertiesquantumrenyi} and it is montonic increasing in $\alpha$ for every $(\rho,\sigma) \in \density \times \density$. Therefore, we can apply the minimax theorem given by~\cite[Corollary A.2]{Mosonyi_2011} and get
\begin{align}
\sup_{\alpha \in (0,1)} D_{\Meas,\alpha}(\sA\|\sB)
& = \sup_{\alpha \in (0,1)} \inf_{\substack{\rho\in\sA\\\sigma\in\sB}} D_{\Meas,\alpha}(\rho\|\sigma)\\
& = \inf_{\substack{\rho\in\sA\\\sigma\in\sB}} \sup_{\alpha \in (0,1)}  D_{\Meas,\alpha}(\rho\|\sigma)\\
& = \inf_{\substack{\rho\in\sA\\\sigma\in\sB}} D_{\Meas}(\rho\|\sigma)\\
& = D_{\Meas}(\sA\|\sB),
\end{align}
where the third line follows from the continuity of $D_{\Meas,\alpha}$ at $\alpha = 1$ {between a state and a positive semidefinite operator (see e.g.~\cite{mosonyi2023continuitypropertiesquantumrenyi})}.    
\end{proof}

Combining Lemma~\ref{lem: generalized supadditivity DM alpha} and Lemma~\ref{lem: general DM alpha continuitiy}, we have the superadditivity for measured relative entropy as follows.

\begin{shaded}
    \begin{lemma}\label{lem: generalized supadditivity DM}
    Let $\cH_{1}, \cH_{2}$ be finite-dimensional Hilbert spaces. Consider convex and compact sets $\sA_1 \subseteq \density(\cH_1)$, $\sA_2 \subseteq \density(\cH_2)$, and $\sA_{12} \subseteq \density(\cH_1 \otimes \cH_2)$ and also $\sB_1 \subseteq \PSD(\cH_1)$, $\sB_{2} \subseteq \PSD(\cH_2)$, and $\sB_{12} \subseteq \PSD(\cH_{1} \otimes \cH_{2})$. Assume that $\polarPSD{(\sA_{1})} \ox \polarPSD{(\sA_{2})} \subseteq \polarPSD{(\sA_{12})}$ and $\polarPSD{(\sB_1)} \ox \polarPSD{(\sB_2)} \subseteq \polarPSD{(\sB_{12})}$. Then, it holds that
    \begin{align}
        D_{\Meas}(\sA_{12}\|\sB_{12}) \geq D_{\Meas}(\sA_1\|\sB_1) + D_{\Meas}(\sA_2\|\sB_2). 
    \end{align}
    \end{lemma}
\end{shaded}

\begin{remark}
The above superadditivity is different than the result by Piani~\cite[Theorem 1]{Piani2009} under compatible assumptions. Here we do not put constrains on the performed measurements. Therefore, the measurements and the set of quantum states may not be compatible in the case here.
\end{remark}

\section{Generalized quantum AEP}
\label{sec: Generalized asymptotic equipartition property}

The asymptotic equipartition property (AEP) is a fundamental concept in information theory that describes the behavior of sequences of random variables as the sequence length increases~\cite{cover1999elements,holenstein2011randomness}. It essentially states that, for a large number of independent and identically distributed (i.i.d.) random variables, the sequences exhibit regular and predictable behavior when considered collectively, despite individual randomness. This concept has been extended to quantum information theory, where the quantum version of the AEP applies to quantum states and quantum entropy (e.g.~\cite{tomamichel2009fully,Tomamichel2015b}). This extension is particularly useful in the study of quantum source coding and other quantum communication protocols. In this work, we further generalize the quantum AEP beyond the i.i.d. framework by considering quantum states drawn from two sets that satisfy the following assumptions.

\begin{assumption}\label{ass: steins lemma assumptions}
Consider a family of sets $\{\sA_n\}_{n\in \NN}$ satisfying the following properties,
\begin{itemize}
    \item (A.1) Each $\sA_n$ is convex and compact;
    \item (A.2) Each $\sA_n$ is permutation-invariant; 
    \item (A.3) $\sA_m \ox \sA_k \subseteq \sA_{m+k}$, for all $m,k \in \NN$;
    \item (A.4) $\polarPSD{(\sA_m)} \ox \polarPSD{(\sA_k)} \subseteq \polarPSD{(\sA_{m+k})}$, for all $m,k \in \NN$.
\end{itemize}
\end{assumption}

The first three assumptions are standard in much of the existing literature (e.g.~\cite{Brand_o_2010,brandao2020adversarial,hayashi2024generalized,lami2024solutiongeneralisedquantumsteins}). The additional polar assumption (A.4) is also satisfied in many cases of interest; see Table~\ref{tab: A list of sets that satisfy Assumption} below. Moreover, these assumptions imply the partial-trace condition used in~\cite{Brand_o_2010,lami2024solutiongeneralisedquantumsteins}; see Appendix~\ref{app: partial trace assumption} for details.

\setlength\extrarowheight{4pt}
\begin{table}[H]
\centering
\begin{tabular}{l|l}
\toprule[2pt]
Sets & Mathematical descriptions \\
\hline
Singleton & $\{\rho^{\ox n}\}$ with $\rho \in \density(\cH)$\\
Conditional states & $\{I_n \ox \rho_n: \rho_n \in \density(\cH^{\ox n})\}$\\
Channel image & $\{\cN^{\ox n}(\rho_n): \rho_n \in \density(\cH^n)\}$ with a quantum channel $\cN$\\
Recovery set & $\{\cN_{B^n\to C^n}(\rho_{AB}^{\ox n}): \cN \in \CPTP(B^n:C^n)\}$ with $\rho \in \density(AB)$\\
Extensions set & $\{\omega_n \in \density(A^nB^n): \tr_{B^n} \omega_n = \rho_{A}^{\ox n}\}$ with $\rho_A \in \density(A)$\\
Incoherent states & $\{\rho_n \in \density(\cH^{\ox n}): \rho_n = \Delta(\rho_n)\}$ with the completely dephasing channel $\Delta$ \\
Rains set &  $\{\rho_n \in \PSD(A^nB^n): \|\rho_n^{\sfT_{B_1\cdots B_n}}\|_1 \leq 1\}$ with the partial transpose $\sfT_{B_i}$ \\
Nonpositive mana & $\{\rho_n \in \PSD(\cH^{\ox n}): \|\rho_n\|_{W,1} \leq 1\}$ with the Wigner trace norm$\|\cdot\|_{W,1}$\\
\bottomrule[2pt]  
\end{tabular}
\caption{\small A list of sets that satisfy Assumption~\ref{ass: steins lemma assumptions}.}
\label{tab: A list of sets that satisfy Assumption}
\end{table}

For instance, the assumptions hold when the set is a singleton of tensor product i.i.d. states, as in the existing quantum AEP; when it consists of the identity operator tensored with the set of all density operators, which is relevant for conditional quantum entropies; or when it is the image of a quantum channel or a set of quantum channels, a scenario that naturally arises in adversarial quantum channel discrimination. The assumptions are also satisfied by the set of recovery states used in quantum state redistribution~\cite{brandao2015quantum,fawzi2015quantum,berta2016fidelity,Berta2017} and the extensions set used in Uhlmann's theorem~\cite{uhlmann1976transition,Mazzola2025,fang2025variational} and the set of incoherent states used in quantum coherence theory~\cite{baumgratz2014quantifying,Winter2016,regula2018one,diaz2018using,fang2018probabilistic,hayashi2021finite}. 

It is also worth mentioning that even in cases where assumption (A.4) is not directly met, we can relax the original set to fulfill the required conditions. For example, the set of separable states does not satisfy the polar assumption, but we can relax it to the Rains set~\cite{rains2001semidefinite,audenaert2002asymptotic}, which is widely used in entanglement theory and satisfies all required assumptions. We can also relax the set of stabilizer states to those with non-positive mana~\cite{Wang2018magicstates}, a technique used in fault-tolerant quantum computing. Exploiting this idea, we can obtain efficient and improved bounds for quantifying quantum resources in several tasks of interest~\cite{fang2025efficient}.

Now we proceed to present our main result in this work as follows.

\begin{shaded}
\begin{theorem}(Generalized AEP.)\label{thm: generalized AEP}
Let $\{\sA_n\}_{n\in\NN}$ and $\{\sB_n\}_{n\in\NN}$ be two sequences of sets satisfying Assumption~\ref{ass: steins lemma assumptions} and $\sA_n \subseteq \density(\cH^{\ox n})$, $\sB_n \subseteq \PSD(\cH^{\ox n})$. Let $d = \dim \cH$. 
Assume moreover that $D_{\max}(\sA_n\|\sB_n) \leq cn$, for all $n \in \NN$ and a constant $c \in \RR_{\pl}$.
Then for any $\ve \in (0,1)$, it holds that
\begin{align}
\label{eq: asymptotic AEP}
    \lim_{n \to \infty} \frac{1}{n}D_{\Hypo,\ve}(\sA_n\|\sB_n) = \lim_{n \to \infty} \frac{1}{n} D_{\max, \ve}(\sA_n \| \sB_n) = D^{\reg}(\sA \| \sB).
\end{align}
In addition, $D^{\reg}(\sA \| \sB)$ can be estimated using the following bounds: for any $m \geq 1$, we have
\begin{align}\label{eq: generalized AEP finite estimate}
    &\frac{1}{m} D_{\Meas}(\sA_m \| \sB_m)\leq D^{\reg}(\sA \| \sB) \leq \frac{1}{m} D(\sA_m \| \sB_m)
\end{align}
with explicit convergence guarantees
\begin{align}\label{eq: generalized AEP finite estimate convergence}
        \frac{1}{m} D(\sA_m \| \sB_m) - \frac{1}{m} D_{\Meas}(\sA_m \| \sB_m) \leq 
        \frac{1}{m} 2(d^2 + d) \log (m+d).
\end{align}

If there is a constant $C$ such that for any $n \geq 1$ and $\rho_n \in \sA_n, \sigma_n \in \sB_n$, $D_{\max}(\rho_n \| \sigma_n) \leq Cn/4$ and $\log \tr(\sigma_n) \leq Cn/4$,\footnotemark we can obtain explicit bounds of the form:
\begin{align}\label{eq: hyp testing rel entropy explicit}
    - C' f(n, \ve)
    \leq D_{\Hypo,\ve}(\sA_n\|\sB_n) - n D^{\reg}(\sA \| \sB)  \leq C' f(n, 1-\ve),
\end{align}
 and 
\begin{align}
      - C'f(n,1-2\ve) - \log \frac{1} {\ve}
      \leq D_{\max,\ve}(\sA_n\|\sB_n) - n D^{\reg} (\sA \| \sB)  
      \leq C' f(n, \ve) + \log \frac{2}{\ve^2},\label{eq: dmax explicit}
\end{align}
where $f(n, \ve) = n^{2/3}\log n \log^{1/3} \frac{1}{\ve}$ and $C'$ only depends on $C$ and the local dimension $d$.
\end{theorem}
\end{shaded}
\footnotetext{{For simplicity of exposition, we stated a rather strong condition that requires $D_{\max}(\rho_n \| \sigma_n) \leq Cn$ for all $\rho_n \in \sA_n, \sigma_n \in \sB_n$. This condition is satisfied for example if $\sigma_n \geq \lambda^n I_n$ for some $\lambda > 0$. A weaker but more technical condition is sufficient for the statement to hold, see~\eqref{eq: assumptionC}. This weaker condition only requires the bound to hold for pairs $\rho_n, \sigma_n$ achieving $D_{\Petz,\alpha}(\sA_n \| \sB_n)$. As shown in~\cite{fang2025adversarial}, this weaker condition is satisfied when the divergences correspond to optimized conditional entropies}.}

\begin{shaded}
\begin{remark}(Approximating $D^{\reg}(\sA \| \sB) $ efficiently.)
\label{rem:efficient computation}
    Eq.~\eqref{eq: generalized AEP finite estimate} provides converging upper and lower bounds on $D^{\reg}(\sA \| \sB)$ in terms of convex optimization programs as both $D$ and $D_{\Meas}$ are convex. In addition, if $\sA_m$ and $\sB_m$ are semidefinite representable (i.e., can be written as the feasible sets of a semidefinite program), then the upper bound $\frac{1}{m} D(\sA_m \| \sB_m)$ is a quantum relative entropy program~\cite{chandrasekaran2017relative,fawzi2018efficient} that can be solved using interior point methods~\cite{fawzi2023optimal}. For the lower bound $\frac{1}{m} D_{\Meas}(\sA_m \| \sB_m)$, we use Lemma~\ref{lem: DM variational sets} to write
    \begin{align}
        D_{\Meas}(\sA_m \| \sB_m) 
        &= 
        \sup_{W, X} \left\{ -h_{\sA_m}(X) :  W\in \polarPD{(\sB_m)}, X \geq -\log W \right\}\\
        &=\sup_{W,X,t,s} \left\{ 
        -t+s : t > 0, \; W \in \polarPD{(\sB_m)}, X + sI \in t\polarPSD{(\sA_m)}, X \geq -\log W  \right\}  \label{eq:reformulation DM sets}
    \end{align}
    where we used the fact that $h_{\sA_m}$ is monotone (i.e., $h_{\sA_m}(X) \leq h_{\sA_m}(Y)$ if $X\leq Y$) since $\sA_m \subseteq \PSD$, and that $h_{\sA_m}(X + sI) = h_{\sA_m}(X) + s$ as $\sA_m \subseteq \density$ so that $h_{\sA_m}(X) = \inf_{s} \{h_{\sA_m}(X + sI) - s : X + sI \geq 0\} = \inf\{ t - s : t > 0 \text{ and } X + sI \in t{\sA_m^\circ}, X+sI \geq 0 \}$.
    If $\sA_m$ and $\sB_m$ are semidefinite representable, then, using semidefinite programming duality, we know that the sets $(\sB_m)^{\circ}_{\pl}$ and
    $\{(t,s,X) : t > 0 \text{ and } X+sI \in t \polarPSD{(\sA_m)}\}$
    are semidefinite representable. The last constraint in \eqref{eq:reformulation DM sets} is convex since the matrix logarithm is operator concave, and the whole maximization problem \eqref{eq:reformulation DM sets} can also be solved using interior-point methods~\cite{fawzi2023optimal}. Further details on this computational aspect are provided in the accompanying paper~\cite{fang2025efficient}.
\end{remark}
\end{shaded}

The proof of Theorem~\ref{thm: generalized AEP} is provided in the following section.

\subsection{Lemmas required to prove the generalized AEP}

The proof of Theorem~\ref{thm: generalized AEP} requires the following lemmas. 

\begin{shaded}
\begin{lemma}(Subadditivity.)\label{lem: generalized subadditivity}
Let $\DD$ be a quantum divergence that is additive or subadditive under tensor product of quantum states. Let $\{\sA_n\}_{n\in\NN}$ and $\{\sB_n\}_{n\in\NN}$ be two sequences of sets satisfying (A.1) and (A.3) in Assumption~\ref{ass: steins lemma assumptions} and $\sA_n \subseteq \density(\cH^{\ox n})$, $\sB_n \subseteq \PSD(\cH^{\ox n})$, and $\DD(\sA_n\|\sB_n) \leq c n$ for all $n \in \NN$ and a constant $c\in \RR_{\pl}$. Then for any $m,k \in \NN$, it holds that
\begin{align}
    \DD(\sA_{m+k}\|\sB_{m+k}) \leq \DD(\sA_m\|\sB_{m}) + \DD(\sA_{k}\|\sB_{k}). 
\end{align}
\end{lemma}
\end{shaded}
\begin{proof}
Note that the assumption $\DD(\sA_n\|\sB_n) \leq c n$ ensures that all divergences considered here are finite. For any $\rho_m \in \sA_m$, $\sigma_m \in \sB_m$, $\rho_k \in \sA_k$ and $\sigma_k \in \sB_k$, we know that $\rho_m \ox \rho_k \in \sA_{m+k}$ and $\sigma_m \ox \sigma_k \in \sB_{m+k}$ by the assumption (A.3). This gives
\begin{align}
    \DD(\sA_{m+k}\|\sB_{m+k}) \leq \DD(\rho_m \ox \rho_k\|\sigma_m \ox \sigma_k) \leq \DD(\rho_m\|\sigma_m) + \DD(\rho_k\|\sigma_k),
\end{align}
where the second inequality follows by the additivity or subadditivity assumption of $\DD$.
As the above holds for any $\rho_m,\sigma_m,\rho_k,\sigma_k$, we have
\begin{align}
    \DD(\sA_{m+k}\|\sB_{m+k}) \leq \inf_{\substack{\rho_m \in \sA_m\\ \sigma_m \in \sB_m}} \DD(\rho_m\|\sigma_m) + \inf_{\substack{\rho_k \in \sA_k\\ \sigma_k \in \sB_k}} \DD(\rho_k\|\sigma_k) = \DD(\sA_m\|\sB_m) + \DD(\sA_k\|\sB_k),
\end{align}
which completes the proof.
\end{proof}

\begin{remark}
Note that $D_{\Meas,1/2}(\rho\|\sigma) = D_{\Sand,1/2}(\rho\|\sigma) = - \log {[F(\rho\|\sigma)]^2}$. Since the recovery set satisfies Assumption~\ref{ass: steins lemma assumptions}, Lemmas~\ref{lem: generalized supadditivity DM alpha} and~\ref{lem: generalized subadditivity} together imply that the fidelity of recovery is multiplicative, as previously proven in~\cite{berta2016fidelity}.
\end{remark}

{By Fekete's lemma~\cite{fekete1923distribution}}, if a sequence $\{a_n\}_{n\in\NN}$ satisfies $a_{n+m} \leq a_n + a_m$, then $a_n/n$ converges. Therefore, if $\DD$ satisfies the above assumptions, we have
\begin{align}\label{eq: general subadditivity limit exist}
    \DD^\reg(\sA\|\sB) = \inf_{n>1} \frac{1}{n} \DD(\sA_n\|\sB_n).
\end{align}

\begin{shaded}
\begin{lemma}(Continuity.)\label{lem: composite sandwiched regularized continuitiy}
Let $\{\sA_n\}_{n\in\NN}$ and $\{\sB_n\}_{n\in\NN}$ be two sequences of sets such that $\sA_n \subseteq \density(\cH^{\ox n})$, $\sB_n \subseteq \PSD(\cH^{\ox n})$ and $D_{\max}(\sA_n\|\sB_n) \leq cn$, for all $n \in \NN$ and a constant $c \in \RR_{\pl}$. If these sets satisfy (A.1), (A.3) in  Assumption~\ref{ass: steins lemma assumptions}, it holds that
\begin{alignat}{2}
    \inf_{\alpha > 1} \; \; & D_{\Sand,\alpha}^{\reg}(\sA\|\sB) & =  D^\reg(\sA\|\sB).
\end{alignat}
If these sets satisfy (A.1), (A.4) in  Assumption~\ref{ass: steins lemma assumptions}, it holds that
\begin{alignat}{2}
    \sup_{\alpha \in (0,1)} & D_{\Meas,\alpha}^\reg(\sA\|\sB) & = D_{\Meas}^\reg(\sA\|\sB).
\end{alignat}
\end{lemma}
\end{shaded}
\begin{proof}
Note that the assumption on $D_{\max}(\sA_n\|\sB_n)$ ensures that all the divergences considered here are finite. Then the following chain of relations hold
\begin{align}
    \inf_{\alpha > 1} D_{\Sand,\alpha}^{\reg}(\sA\|\sB) & = \inf_{\alpha > 1} \inf_{n>1} \frac{1}{n} D_{\Sand,\alpha}(\sA_n\|\sB_n)\\
& = \inf_{\alpha > 1} \inf_{n>1}  \inf_{\substack{\rho_n \in \sA_n\\ \sigma_n \in \sB_n}}\frac{1}{n} D_{\Sand,\alpha}(\rho_n\|\sigma_n)\\
& = \inf_{n>1}  \inf_{\substack{\rho_n \in \sA_n\\ \sigma_n \in \sB_n}} \inf_{\alpha > 1} \frac{1}{n} D_{\Sand,\alpha}(\rho_n\|\sigma_n)\\\
& = \inf_{n>1}  \inf_{\substack{\rho_n \in \sA_n\\ \sigma_n \in \sB_n}} \frac{1}{n} D(\rho_n\|\sigma_n)\\
& =  \inf_{n>1} \frac{1}{n} D(\sA_n\|\sB_n)\\
& = D^\reg(\sA\|\sB)
\end{align}
where the first and the last equalities follow from the assumptions (A.1) and (A.3) and therefore the subadditivity of $D_{\Sand,\alpha}(\sA_n\|\sB_n)$ and $D(\sA_n\|\sB_n)$ from Lemma~\ref{lem: generalized subadditivity}, and the fourth equality uses the continuity of the sandwiched \Renyi divergence at $\alpha = 1$ {between a state and a positive semidefinite operator (see e.g.~\cite{lin2015investigating})}. 

Similarly, we have the following 
\begin{align}
    \sup_{\alpha \in (0,1)} D_{\Meas,\alpha}^\reg(\sA\|\sB) & = \sup_{\alpha \in (0,1)} \sup_{n > 1} \frac{1}{n} D_{\Meas,\alpha}(\sA_n\|\sB_n)\\
    & = \sup_{n > 1} \sup_{\alpha \in (0,1)}\frac{1}{n} D_{\Meas,\alpha}(\sA_n\|\sB_n)\\
    & = \sup_{n > 1} \frac{1}{n} D_{\Meas}(\sA_n\|\sB_n)\\
    & = D_{\Meas}^\reg(\sA\|\sB),
\end{align}
where the first and the last equalities follow from the assumptions (A.1) and (A.4) and therefore the superadditivity from Lemma~\ref{lem: generalized supadditivity DM alpha} and Lemma~\ref{lem: generalized supadditivity DM}, and the third equality uses the continuity of $D_{\Meas,\alpha}(\sA_n\|\sB_n)$ at $\alpha = 1$ from Lemma~\ref{lem: general DM alpha continuitiy}. This completes the proof.
\end{proof}

\begin{shaded}
\begin{lemma}(Asymptotic equivalence.)\label{lem: DM alpha DS alpha finite relation}
Let $\{\sA_n\}_{n\in\NN}$ and $\{\sB_n\}_{n\in\NN}$ be two sequences of sets satisfying (A.1) and (A.2) in Assumption~\ref{ass: steins lemma assumptions} and $\sA_n \subseteq \density(\cH^{\ox n})$, $\sB_n \subseteq \PSD(\cH^{\ox n})$ and $D_{\max}(\sA_n\|\sB_n) \leq cn$, for all $n \in \NN$ and a constant $c \in \RR_{\pl}$.  Let $d = \dim(\cH)$ and $\alpha \in [1/2,\infty)$. Then it holds that
\begin{align}\label{eq: DM alpha DS alpha finite relation}
\frac{1}{n} D_{\Meas,\alpha}(\sA_n\|\sB_n) \leq \frac{1}{n} D_{\Sand,\alpha}(\sA_n\|\sB_n) \leq \frac{1}{n} D_{\Meas,\alpha}(\sA_n\|\sB_n) + \frac{1}{n} 2(d^2 + d) \log (n+d).
\end{align}
As a consequence, it holds that
\begin{align}\label{eq: DM alpha DS alpha regularization}
D_{\Meas,\alpha}^\reg(\sA\|\sB) = D_{\Sand,\alpha}^\reg(\sA\|\sB).
\end{align}
\end{lemma}
\end{shaded}
\begin{proof}
Since $D_{\Meas,\alpha}(\rho\|\sigma) \leq D_{\Sand,\alpha}(\rho\|\sigma)$, it is clear that
\begin{align}\label{eq: DM alpha D alpha tmp1}
    \frac{1}{n} D_{\Meas,\alpha}(\sA_n\|\sB_n) \leq \frac{1}{n} D_{\Sand,\alpha}(\sA_n\|\sB_n).
\end{align}
Now we prove the second inequality. Let $\cT$ be the twirling operation under permutation group. Then for any $\rho_n \in \sA_n$ and $\sigma_n \in \sB_n$, we have $\cT(\rho_n) \in \sA_n$ and $\cT(\sigma_n) \in \sB_n$ by the permutation-invariant assumption (A.2) of $\{\sA_n\}_{n\in \NN}$ and $\{\sB_n\}_{n\in \NN}$. This gives
\begin{align}
D_{\Sand,\alpha}(\sA_n\|\sB_n) & \leq D_{\Sand,\alpha}(\cT(\rho_n)\|\cT(\sigma_n))\\
    & \leq D_{\Meas,\alpha}(\cT(\rho_n)\|\cT(\sigma_n)) + 2 \log |\spec(\cT(\sigma_n))|\\
    & \leq D_{\Meas,\alpha}(\cT(\rho_n)\|\cT(\sigma_n)) + 2 \log \left[(n+1)^d(n+d)^{d^2}\right]\\
    & \leq D_{\Meas,\alpha}(\rho_n\|\sigma_n) + 2 \log \left[(n+1)^d(n+d)^{d^2}\right],
\end{align}
where the second inequality follows from Lemma~\ref{lem: DM and Sandwiched relation}, the third inequality follows from Lemma~\ref{lem: permutation invariant spec} and the fact that $\cT(\sigma_n)$ is permutation invariant, the last inequality follows from the data-processing inequality of $D_{\Meas}$. As the above holds for any $\rho_n \in \sA_n$ and $\sigma_n \in \sB_n$, we have
\begin{align}
    \frac{1}{n}D_{\Sand,\alpha}(\sA_n\|\sB_n) \leq \frac{1}{n} D_{\Meas,\alpha}(\sA_n\|\sB_n) + \frac{2}{n} \log \left[(n+1)^d(n+d)^{d^2}\right].
\end{align} 
Finally, noting that $\log [(n+1)^d(n+d)^{d^2}] \leq (d^2 + d) \log (n+d)$, we have the result in Eq.~\eqref{eq: DM alpha DS alpha finite relation}. Taking $\lim_{n\to \infty}$ in Eq.~\eqref{eq: DM alpha DS alpha finite relation}, we have the relation in Eq.~\eqref{eq: DM alpha DS alpha regularization}.
\end{proof}

\begin{shaded}
\begin{lemma}(Finite estimation.)\label{lem: main theorem computation}
    Let $\{\sA_n\}_{n\in\NN}$ and $\{\sB_n\}_{n\in\NN}$ be two sequences of sets satisfying Assumption~\ref{ass: steins lemma assumptions} and $\sA_n \subseteq \density(\cH^{\ox n})$, $\sB_n \subseteq \PSD(\cH^{\ox n})$ and $D_{\max}(\sA_n\|\sB_n) \leq cn$, for all $n \in \NN$ and a constant $c \in \RR_{\pl}$. Then it holds that
\begin{align}\label{eq: main theorem computation tmp1}
    \frac{1}{m} D_{\Meas}(\sA_m\|\sB_m) \leq D^\reg(\sA\|\sB) \leq \frac{1}{m} D(\sA_m \| \sB_m).
\end{align}
\end{lemma}
\end{shaded}
\begin{proof}
By the superadditivity of the measured relative entropy in Lemma~\ref{lem: generalized supadditivity DM} and the subadditivity of the Umegaki relative entropy in Lemma~\ref{lem: generalized subadditivity} (by choosing $\alpha = 1$), we have 
\begin{align}
    D_{\Meas}^{\reg}(\sA\|\sB) & = \sup_{m > 1} \frac{1}{m} D_{\Meas}(\sA_m\|\sB_m),\\
    D^{\reg}(\sA\|\sB) & = \inf_{m \geq 1} \frac{1}{m} \;\; D(\sA_m\|\sB_m).
\end{align}
This gives
\begin{align}
    \frac{1}{m} D_{\Meas}(\sA_m\|\sB_m) \leq  D_{\Meas}^{\reg}(\sA\|\sB) = D^{\reg}(\sA\|\sB) \leq \frac{1}{m} D(\sA_m\|\sB_m),
\end{align}
where the equality follows from the asymptotic equivalence in Eq.~\eqref{eq: DM alpha DS alpha regularization} by choosing $\alpha = 1$.
\end{proof}

In order to obtain the explicit error bound, we use the following lemma.
\begin{shaded}
\begin{lemma}
\label{eq: approx Dreg DMmalpha}
Let $\{\sA_n\}_{n\in\NN}$ and $\{\sB_n\}_{n\in\NN}$ be two sequences of sets satisfying Assumption~\ref{ass: steins lemma assumptions} and $\sA_n \subseteq \density(\cH^{\ox n})$, $\sB_n \subseteq \PSD(\cH^{\ox n})$. Let $m \geq 2$ and assume there is a constant $C > 0$ such that 
\begin{align}
\tag{$*$}
\label{eq: assumptionC}
     \forall \alpha \in [1/2,1]: \; D_{\Petz,3/2}(\rho^{(\alpha)}_m \|\sigma^{(\alpha)}_m) \leq \frac{C}{4} m \quad \text{ and } \quad \forall \sigma_m \in \sB_m : \; \log \tr(\sigma_m) \leq \frac{C}{4} m, 
\end{align}
where 
$\rho^{(\alpha)}_m \in \sA_m, \sigma^{(\alpha)}_m \in \sB_m$ satisfy $D_{\Petz,\alpha}(\rho_m^{(\alpha)} \| \sigma_m^{(\alpha)}) = D_{\Petz,\alpha}(\sA_m \| \sB_m)$. 

\vspace{0.2cm}
For $1 - \frac{1}{(2+C)m} < \alpha < 1$, we have
\begin{align}
\label{eq: approx Dreg from below explicit}
0 \leq D^{\reg}(\sA \| \sB) - \frac{1}{m} D_{\Meas,\alpha}(\sA_m \| \sB_m) \leq (1-\alpha) (2+C)^2 m + \frac{2(d^2 + d) \log (m+d)}{m},
\end{align}
and for $1 < \alpha < 1 + \frac{1}{(2+C)m}$, we have
\begin{align} 
\label{eq: approx Dreg from above explicit}
0 \leq \frac{1}{m} D_{\Sand, \alpha}(\sA_m \| \sB_m) - D^{\reg}(\sA \| \sB) \leq (\alpha - 1) (2+C)^2 m + \frac{2(d^2+d) \log(m+d)}{m}.
\end{align}
\end{lemma}
\end{shaded}
\begin{proof}
We start by establishing~\eqref{eq: approx Dreg from above explicit}. For that, we use the continuity of $D_{\Sand, \alpha}$ as $\alpha \to 1^+$ as stated in Lemma~\ref{lem: continuity alpha = 1 from above} (Appendix~\ref{sec: app Useful properties}). More specifically,
\begin{align}
    \frac{1}{m} D_{\Sand,\alpha}(\sA_m \| \sB_m)
    &= \frac{1}{m} \inf_{\substack{\rho_m \in \sA_m\\ \sigma_m \in \sB_m}} D_{\Sand, \alpha}(\rho_m \| \sigma_m) \\
    &\leq \frac{1}{m} \inf_{\substack{\rho_m \in \sA_m\\ \sigma_m \in \sB_m}}
    D(\rho_m \| \sigma_m) + (\alpha - 1) (\log \eta_m)^2 \\
    &\leq \frac{1}{m}
    D(\rho_m^{(1)} \| \sigma_m^{(1)}) + (\alpha - 1) (\log \eta^{(1)}_m)^2
\end{align}
for $\alpha \in (1, 1+1/\log \eta_m)$ where $\sqrt{\eta_m} = \max(4, 2^{2D_{\Petz,3/2}(\rho_m \| \sigma_m)} + 2^{-2D_{\Petz, 1/2}(\rho_m \| \sigma_m)} + 1)$ and $\sqrt{\eta_m^{(1)}} =  \max(4, 2^{2D_{\Petz,3/2}(\rho^{(1)}_m \| \sigma^{(1)}_m)} + 2^{-2D_{\Petz, 1/2}(\rho^{(1)}_m \| \sigma^{(1)}_m)} + 1)$. The assumption~\eqref{eq: assumptionC} gives $D_{\Petz,3/2}(\rho^{(1)}_m \| \sigma^{(1)}_m) \leq \frac{C}{4}m$ and $-D_{\Petz, 1/2}(\rho^{(1)}_m \| \sigma^{(1)}_m) \leq \log \tr(\sigma^{(1)}_m) \leq \frac{C}{4} m$. As a result, $\eta^{(1)}_m \leq (3 \cdot 2^{Cm/2})^2$, which implies $\log \eta^{(1)}_m \leq 2 \log 3 + Cm \leq (2+C) m$. Finally, for $\alpha \in \left(1, 1+\frac{1}{(2+C)m}\right)$, we have 
\begin{align}
    \frac{1}{m} D_{\Sand,\alpha}(\sA_m \| \sB_m)
    &\leq \frac{1}{m} D(\sA_m \| \sB_m) + (\alpha - 1) (2+C)^2 m \\
    &\leq \frac{1}{m} D_{\Meas}(\sA \| \sB) + \frac{1}{m}2(d^2+d) \log(m+d) + (\alpha - 1) (2+C)^2m \\
    &\leq D^{\reg}(\sA \| \sB) + \frac{1}{m}2(d^2+d) \log(m+d) + (\alpha - 1) (2+C)^2 m
\end{align}
using Lemma~\ref{lem: DM alpha DS alpha finite relation} then Lemma~\ref{lem: main theorem computation}.

Now, we move to the approximation from below in~\eqref{eq: approx Dreg from below explicit}. We have
\begin{align}
    D^{\reg}(\sA \| \sB) - \frac{1}{m} D_{\Meas,\alpha}(\sA_m \| \sB_m)  
    &\leq \frac{1}{m} D(\sA_m \| \sB_m) - \frac{1}{m} D_{\Sand, \alpha}(\sA_m \| \sB_m) \\
    &\quad + \frac{1}{m} D_{\Sand, \alpha}(\sA_m \| \sB_m) - \frac{1}{m} D_{\Meas,\alpha}(\sA_m \| \sB_m).
\end{align}
For the term $\frac{1}{m} D_{\Sand, \alpha}(\sA_m \| \sB_m) - \frac{1}{m} D_{\Meas,\alpha}(\sA_m \| \sB_m)$, we use Lemma~\ref{lem: DM alpha DS alpha finite relation} and obtain the bound
\begin{align}
    \frac{1}{m} D_{\Sand, \alpha}(\sA_m \| \sB_m) - \frac{1}{m} D_{\Meas,\alpha}(\sA_m \| \sB_m) \leq 2\frac{(d^2+d)\log (m + d)}{m}.
\end{align}
For the term $\frac{1}{m} D(\sA_m \| \sB_m) - \frac{1}{m} D_{\Sand, \alpha}(\sA_m \| \sB_m)$, we use the continuity bound for $\alpha \to 1^-$ in Lemma~\ref{lem: continuity alpha = 1 from below} (Appendix~\ref{sec: app Useful properties}). For $\alpha \in (1/2,1)$, we use the relation between the sandwiched and Petz R\'enyi divergences, 
\begin{align}
    D_{\Sand, \alpha}(\sA_m \| \sB_m)
    &\geq \inf_{\rho_m \in \sA_m, \sigma_m \in \sB_m} \alpha D_{\Petz, \alpha}(\rho_m \| \sigma_m) - (1-\alpha) \log \tr \sigma_m \\
    &\geq D_{\Petz, \alpha}(\rho_m^{(\alpha)} \| \sigma^{(\alpha)}_m) - (1-\alpha)\left(D_{\Petz, \alpha}(\rho_m^{(\alpha)} \| \sigma^{(\alpha)}_m) + \frac{C}{4}m\right) \\
    &\geq D(\rho_m^{(\alpha)} \| \sigma^{(\alpha)}_m) - (1-\alpha) \left((\log \eta^{(\alpha)}_m)^2 + D(\rho^{(\alpha)}_m \| \sigma^{(\alpha)}_m) + \frac{C}{4}m \right),
\end{align}
provided $\alpha \in (1-1/\log \eta^{(\alpha)}_m, 1)$ with $\sqrt{\eta^{(\alpha)}_m} = \max(4, 1+2^{D_{\Petz,3/2}(\rho^{(\alpha)}_m \| \sigma^{(\alpha)}_m)} + 2^{-D_{\Petz,1/2}(\rho^{(\alpha)}_m \| \sigma^{(\alpha)}_m)})$. {The first line above follows from~\cite[Corollary 2.3]{iten2016pretty}, the second line follows from the assumption~\eqref{eq: assumptionC} and the last line follows from Lemma~\ref{lem: continuity alpha = 1 from below} in Appendix~\ref{sec: app Useful properties}.} Again we have $\log \eta^{(\alpha)}_m \leq 2\log 3 + Cm$. As a result, we get
\begin{align}
    \frac{1}{m} D_{\Sand, \alpha}(\sA_m \| \sB_m) 
    &\geq \frac{1}{m} D(\sA_m \| \sB_m) - (1-\alpha) \left((2\log 3+Cm)^2/m  + C/2 \right) \\
    &\geq \frac{1}{m} D(\sA_m \| \sB_m) - (1-\alpha) (2+C)^2 m,
\end{align}
for $m \geq 2$ and where we used Lemma~\ref{lem: main theorem computation}. 
\end{proof}

Now we are ready to prove the generalized AEP in Theorem~\ref{thm: generalized AEP}.

\subsection{Proof of the generalized AEP}

We prove the results first for $D_{\Hypo, \ve}$. The results for $D_{\max,\ve}$ will follow easily at the end. 
By using the relation of hypothesis testing relative entropy and the sandwiched \Renyi divergence (see Lemma~\ref{lem: DH petz sandwiched}), it holds that for any $\ve \in (0,1)$ and $\alpha > 1$,
\begin{align}
    D_{\Hypo,\ve}(\rho\|\sigma) \leq D_{\Sand,\alpha}(\rho\|\sigma) + \frac{\alpha}{\alpha-1} \log \frac{1}{1-\ve}.
\end{align}
This implies 
\begin{align}
\label{eq: hypo to sand in proof AEP}
    \frac{1}{n} D_{\Hypo,\ve}(\sA_n\|\sB_n) \leq \frac{1}{n} D_{\Sand,\alpha}(\sA_n\|\sB_n) + \frac{1}{n}  \frac{\alpha}{\alpha-1} \log \frac{1}{1-\ve}.
\end{align}
We start by establishing the upper bound on the asymptotic statement~\eqref{eq: asymptotic AEP} for $D_{\Hypo, \ve}$. Taking $n$ to infinity, we have
\begin{align}
    \limsup_{n\to \infty} \frac{1}{n} D_{\Hypo,\ve}(\sA_n\|\sB_n) \leq \limsup_{n\to \infty}\frac{1}{n} D_{\Sand,\alpha}(\sA_n\|\sB_n) = D_{\Sand,\alpha}^\reg(\sA\|\sB).
\end{align}
As the above relation holds for any $\alpha >1$, we can take the infimum of $\alpha$ on the right-hand side and get
\begin{align}
\limsup_{n\to \infty} \frac{1}{n} D_{\Hypo,\ve}(\sA_n\|\sB_n) & \leq \inf_{\alpha > 1} D_{\Sand,\alpha}^\reg(\sA\|\sB) = D^\reg(\sA\|\sB),\label{eq: proof generalized AEP tmp1}
\end{align}
where the equality follows from the continuity in Lemma~\ref{lem: composite sandwiched regularized continuitiy}. 

We now establish the explicit upper bound~\eqref{eq: hyp testing rel entropy explicit}. We start from~\eqref{eq: hypo to sand in proof AEP} and first use the subadditivity of $D_{\Sand, \alpha}$ (Lemma~\ref{lem: generalized subadditivity}) and write for some $m \geq 2$ to be chosen later and $n = m \floor{n/m} + r$ with $0\leq r < m$:
\begin{align}
    D_{\Hypo,\ve}(\sA_n\|\sB_n) 
    &\leq \floor{n/m} D_{\Sand,\alpha}(\sA_m \| \sB_m) + D_{\Sand, \alpha}(\sA_r \| \sB_r) + \frac{\alpha}{\alpha - 1} \log\frac{1}{1-\ve} \\
    &\leq \floor{n/m} D_{\Sand,\alpha}(\sA_m \| \sB_m) + C m + \frac{\alpha}{\alpha - 1} \log\frac{1}{1-\ve}.
    \label{eq: AEP proof sandwiched floor}
\end{align}
We now use Lemma~\ref{eq: approx Dreg DMmalpha} for $m \geq 2$ and $\alpha \leq 1 + \frac{1}{(2+C) m}$ to get
\begin{align}
\frac{1}{m} D_{\Sand,\alpha}(\sA_m \| \sB_m) \leq D^{\reg}(\sA \| \sB) + (\alpha - 1)(2+C)^2 m + \frac{2(d^2+d)\log(m+d)}{m}.
\end{align}
Let us now choose $\alpha - 1 = \frac{8d^2 \log m}{(2+C)^2 m^2}$ and assume that $m \geq \max\left(d,\left(\frac{16 d^2}{2+C}\right)^2\right)$ so that the condition $\alpha \leq 1 + \frac{1}{(2+C)m}$ is satisfied and $\log(m+d) \leq 2 \log m$. With this choice, we have
\begin{align}
D_{\Sand,\alpha}(\sA_m \| \sB_m) \leq m D^{\reg}(\sA \| \sB) + 16 d^2 \log m.
\end{align}
Getting back to~\eqref{eq: AEP proof sandwiched floor}, we get 
\begin{align}
    D_{\Hypo,\ve}(\sA_n\|\sB_n) 
    &\leq \floor{n/m} m D^{\reg}(\sA \| \sB) + \floor{n/m} 16d^2 \log m + C m + \frac{2 (2+C)^2 m^2}{8d^2 \log m} \log\frac{1}{1-\ve}.
\end{align}
Note that if $D^{\reg}(\sA \| \sB) \geq 0$, we have
\begin{align}
\floor{n/m} m D^{\reg}(\sA \| \sB) \leq n D^{\reg}(\sA \| \sB) 
\end{align}
and otherwise,
\begin{align}
\floor{n/m} m D^{\reg}(\sA \| \sB) \leq n D^{\reg}(\sA \| \sB) - m D^{\reg}(\sA \| \sB) \leq n D^{\reg}(\sA \| \sB) + C m,
\end{align}
using the condition $\frac{1}{n} \log \tr(\sigma_n) \leq C$ for any $n \geq 1$ and $\sigma_n \in \sB_n$.
As a result,
\begin{align}
    D_{\Hypo,\ve}(\sA_n\|\sB_n) 
    &\leq n D^{\reg}(\sA \| \sB) + \frac{n 16d^2}{m} \log m + 2 C m + \frac{2 (2+C)^2 m^2}{8d^2 \log m} \log\frac{1}{1-\ve}.
\end{align}
We now choose $m = \left(\frac{64 d^4 n }{(2+C)^2 \log \frac{1}{1-\ve} }\right)^{1/3}$ and get 
\begin{align}
    D_{\Hypo,\ve}(\sA_n\|\sB_n) 
    &\leq n D^{\reg}(\sA \| \sB) + C' n^{2/3} \log n \left(\log \frac{1}{1-\ve}\right)^{1/3}.
\end{align}
for some constant $C'$ that only depends on $d$ and $C$. 

Now we prove the other direction starting first with the asymptotic statement~\eqref{eq: asymptotic AEP}. Using the relation of hypothesis testing relative entropy and the Petz \Renyi divergence (see Lemma~\ref{lem: DH petz sandwiched}), it holds that for any $\ve \in (0,1)$ and $\alpha \in (0,1)$,
\begin{align}
D_{\Hypo,\ve}(\rho\|\sigma) & \geq D_{\Petz,\alpha}(\rho\|\sigma) + \frac{\alpha}{\alpha - 1} \log \frac{1}{\ve}\geq D_{\Meas,\alpha}(\rho\|\sigma)  + \frac{\alpha}{\alpha - 1} \log \frac{1}{\ve},
\end{align}
where the second inequality follows from the data-processing inequality of the Petz \Renyi divergence. Applying this to two sets of quantum states, we get
\begin{align}
\label{eq: AEP proof lower bound measured}
   \frac{1}{n} D_{\Hypo,\ve}(\sA_n\|\sB_n) \geq \frac{1}{n} D_{\Meas,\alpha}(\sA_n\|\sB_n)  +  \frac{1}{n}  \frac{\alpha}{\alpha - 1} \log \frac{1}{\ve}.
\end{align}
Taking $n$ to infinity, we get
\begin{align}
\liminf_{n\to \infty}\frac{1}{n} D_{\Hypo,\ve}(\sA_n\|\sB_n) \geq \liminf_{n\to \infty} \frac{1}{n} D_{\Meas,\alpha}(\sA_n\|\sB_n) = D_{\Meas,\alpha}^\reg(\sA\|\sB). 
\end{align}
As the above relation holds for any $\alpha \in (0,1)$, we take the supremum of $\alpha$ and get
\begin{align}
\liminf_{n\to \infty}\frac{1}{n} D_{\Hypo,\ve}(\sA_n\|\sB_n) & \geq \sup_{\alpha \in (0,1)} D_{\Meas,\alpha}^\reg(\sA\|\sB) = D_{\Meas}^\reg(\sA\|\sB) = D^\reg(\sA\|\sB)
 \label{eq: proof generalized AEP tmp2}
\end{align}
where the first equality follows from the continuity in Lemma~\ref{lem: composite sandwiched regularized continuitiy} and the second equality follows from the asymptotic equivalence in Lemma~\ref{lem: DM alpha DS alpha finite relation}. Putting Eqs.~\eqref{eq: proof generalized AEP tmp1} and~\eqref{eq: proof generalized AEP tmp2} together, we have the asserted result in Eq.~\eqref{eq: asymptotic AEP} about the hypothesis testing relative entropy.

We now prove the explicit lower bound statement in~\eqref{eq: hyp testing rel entropy explicit}. Starting from~\eqref{eq: AEP proof lower bound measured} and using the superadditivity  statement in Lemma~\ref{lem: generalized supadditivity DM} with $n = m \floor{n/m} + r$, we have
\begin{align}
\label{eq: AEP proof measured floor}
D_{\Hypo,\ve}(\sA_n\|\sB_n) 
&\geq \floor{n/m} D_{\Meas,\alpha}(\sA_m \| \sB_m) + D_{\Meas,\alpha}(\sA_r \| \sB_r) + \frac{\alpha}{\alpha - 1}\log\frac{1}{\ve} \\
&\geq \floor{n/m} D_{\Meas,\alpha}(\sA_m \| \sB_m) - Cm + \frac{\alpha}{\alpha - 1}\log\frac{1}{\ve}
\end{align}
where we used the fact that $D_{\Meas,\alpha}(\sA_r \| \sB_r) \geq D_{\Petz,0}(\sA_r \| \sB_r) \geq \inf_{\sigma_r \in \sB_r} -\log \tr \sigma_r \geq -Cr$ as per Lemma~\ref{thm: comparison of quantum divergence}.

We now apply Lemma~\ref{eq: approx Dreg DMmalpha} for $m \geq 2$ and $\alpha \geq 1 - \frac{1}{(2+C) m}$ to get
\begin{align}
\frac{1}{m} D_{\Meas,\alpha}(\sA_m \| \sB_m) \geq D^{\reg}(\sA \| \sB) - (1 - \alpha)(2+C)^2 m + \frac{2(d^2+d)\log(m+d)}{m}.
\end{align}
Let us now choose $1-\alpha = \frac{8d^2 \log m}{(2+C)^2 m^2}$ and assume that $m \geq \max\left(d,\left(\frac{16 d^2}{2+C}\right)^2\right)$ so that the condition $\alpha \geq 1 - \frac{1}{(2+C)m}$ is satisfied and $\log(m+d) \leq 2 \log m$. With this choice, we have
\begin{align}
D_{\Meas,\alpha}(\sA_m \| \sB_m) \geq m D^{\reg}(\sA \| \sB) - 16 d^2 \log m.
\end{align}
Getting back to~\eqref{eq: AEP proof measured floor}, we get 
\begin{align}
    D_{\Hypo,\ve}(\sA_n\|\sB_n) 
    &\geq \floor{n/m} m D^{\reg}(\sA \| \sB) - \floor{n/m} 16d^2 \log m + C m - \frac{2 (2+C)^2 m^2}{8d^2 \log m} \log\frac{1}{\ve}.
\end{align}
Note that if $D^{\reg}(\sA \| \sB) \geq 0$, we have
\begin{align}
\floor{n/m} m D^{\reg}(\sA \| \sB) \geq n D^{\reg}(\sA \| \sB) - m D^{\reg}(\sA \| \sB) \geq  n D^{\reg}(\sA \| \sB) - m C
\end{align}
because for all $n \geq 1$, $D(\sA_n \| \sB_n) \geq D_{\Petz,0}(\sA_n \| \sB_n) \geq \inf_{\sigma_n \in \sB_n} -\log \tr(\sigma_n) \geq -Cn$.
Otherwise, if $D^{\reg}(\sA \| \sB) < 0$, then
\begin{align}
\floor{n/m} m D^{\reg}(\sA \| \sB) \geq n D^{\reg}(\sA \| \sB)
\end{align}
As a result,
\begin{align}
    D_{\Hypo,\ve}(\sA_n\|\sB_n) 
    &\geq n D^{\reg}(\sA \| \sB) - \frac{n 16d^2}{m} \log m - 2 C m - \frac{2 (2+C)^2 m^2}{8d^2 \log m} \log\frac{1}{\ve}.
\end{align}
We now choose $m = \left(\frac{64 d^4 n }{(2+C)^2 \log \frac{1}{\ve} }\right)^{1/3}$ and get 
\begin{align}
    D_{\Hypo,\ve}(\sA_n\|\sB_n) 
    &\geq n D^{\reg}(\sA \| \sB) - C' n^{2/3} \log n \left(\log \frac{1}{\ve}\right)^{1/3}.
\end{align}
for some constant $C'$ that only depends on $d$ and $C$. 

The analogous statements for $D_{\max, \ve}$ follow directly from known relations between $D_{\max, \ve}$ and $D_{\Hypo, \ve}$. In fact, we have for any $\ve \in (0,1)$ and $\ve' \in (0,1-\ve)$,
\begin{align}
D_{\Hypo,\ve'}(\rho\|\sigma)+\log\left(1-\ve-\ve'\right) \leq D_{\max,\ve}(\rho\|\sigma)\leq D_{\Hypo,1-\frac{1}{2}\ve^2}(\rho\|\sigma)+\log\left(\frac{2}{\ve^2}\right).
\end{align}
The upper bound is from~\cite[Proposition 4.1]{dupuis2014generalized} and the lower bound from~\cite[Theorem 11]{datta2013smooth}. Taking the infimum over $\rho \in \sA_n$ and $\sigma_n \in \sB_n$, and then the limit as $n \to \infty$, we get the $D_{\max,\ve}$ statement in~\eqref{eq: asymptotic AEP}. Moreover, combining this bound with $\ve' = 1-2\ve$ and~\eqref{eq: hyp testing rel entropy explicit} gives~\eqref{eq: dmax explicit}.

\section{{Quantum Stein's lemma with composite hypotheses}}
\label{sec: Application 1: another generalized quantum Stein's lemma}

Given many i.i.d. copies of a quantum system described by either the state $\rho$ or $\sigma$ (referred to as the null and alternative hypotheses, respectively), what is the optimal measurement to determine the true state? In asymmetric hypothesis testing, the objective is to minimize the probability of mistakenly identifying $\rho$ as $\sigma$, while keeping the probability of identifying $\sigma$ as $\rho$ below a small, fixed threshold. The well-known quantum Stein’s lemma~\cite{hiai1991proper,Ogawa2000} establishes that the asymptotic exponential rate at which this error probability tends to zero is given by the quantum relative entropy between $\rho$ and $\sigma$. This result was later extended to distinguish between an i.i.d. state and a set of quantum states by Brandão and Plenio~\cite{Brand_o_2010}. However, a gap in the original proof was discovered~\cite{fang2021ultimate,berta2023gap,berta2024tangled}, which was recently resolved by~\cite{hayashi2024generalized}, with an alternative proof provided by~\cite{lami2024solutiongeneralisedquantumsteins}.

In this section, we propose a new generalization of the quantum Stein's lemma that distinguishes between two sets of quantum states beyond the i.i.d. framework. Our results allows both a composite null hypothesis and composite alternate hypothesis.

\subsection{Operational setting: quantum hypothesis testing between two sets of states}

Let us first introduce the operational setting for quantum hypothesis testing between two sets of quantum states. In this setting, a tester draws samples from two sets of quantum states, $\sA_n$ and $\sB_n$, and performs measurements to determine which set the sample belongs to. The null hypothesis assumes that the sample comes from $\sA_n$, while the alternative hypothesis assumes it is drawn from $\sB_n$. An illustration of this task is provided in Figure~\ref{fig: hypothesis testing between two sets}. Such a setting provides a very general framework. For instance, it includes the standard quantum Stein's lemma when $\sA_n$ and $\sB_n$ are two singletons of i.i.d. states~\cite{hiai1991proper,Ogawa2000}. It also encompasses the generalized quantum Stein's lemma~\cite{hayashi2024generalized,lami2024solutiongeneralisedquantumsteins}, where the null hypothesis $\sA_n$ is a singleton of i.i.d. states {(or the set of ``almost power states''~\cite{lami2024solutiongeneralisedquantumsteins})} and $\sB_n$ is a set of quantum states. Additionally, it covers the composite quantum Stein's lemma setting, as explored in~\cite{berta2021composite}, which deals with mixtures of i.i.d. states in both $\sA_n$ and $\sB_n$. This framework can further accommodate the correlated states setting, as studied in~\cite{hiai2008error,mosonyi2015two}, where $\sA_n$ and $\sB_n$ are particular sets of correlated states on a finite spin chain.

\begin{figure}[H]
    \centering
    \includegraphics[width=12cm]{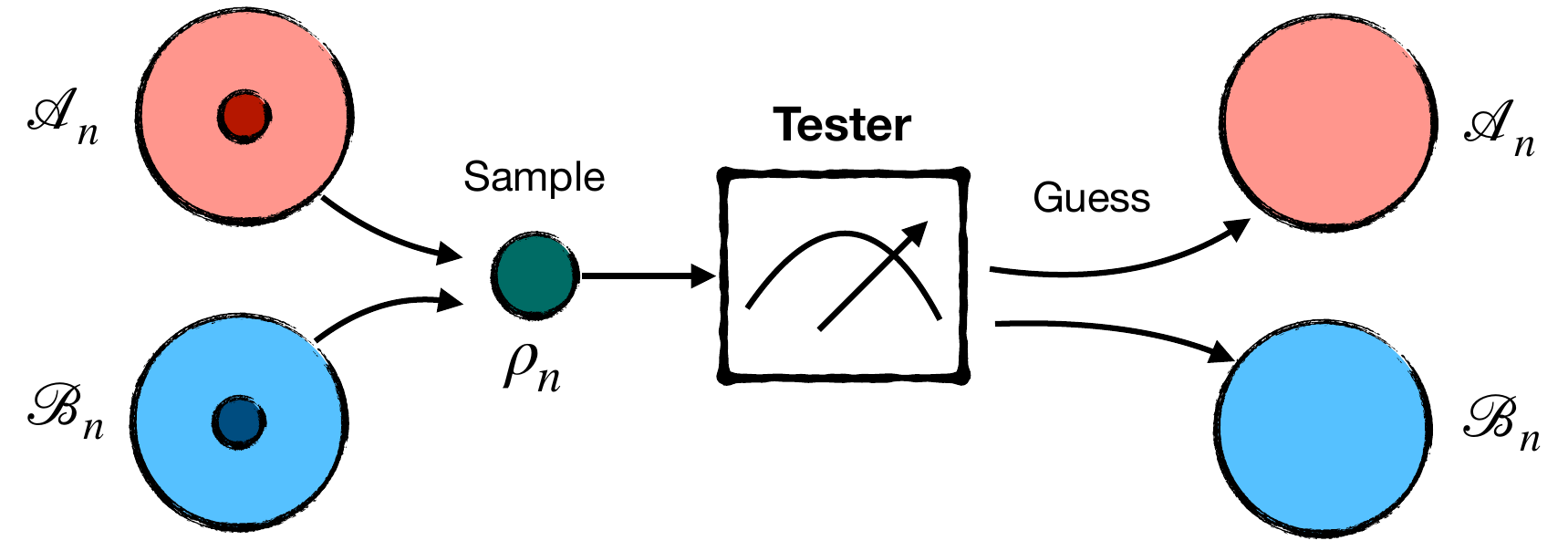}  
    \caption{\small An illustration of the quantum hypothesis testing between two sets of quantum states $\sA_n$ and $\sB_n$. A tester draws samples from the sets and performs quantum measurement to determine which class the sample belongs to. The null hypothesis is that the sample is from $\sA_n$, which typically represents the bad case, and the alternative hypothesis is that the sample is from $\sB_n$.}
    \label{fig: hypothesis testing between two sets}
\end{figure}

As in standard hypothesis testing, the tester will make two types of errors:

\begin{itemize}
    \item \textbf{Type-I error}: The sample is from $\sA_n$, but the tester incorrectly classifies it as from $\sB_n$,
    \item \textbf{Type-II error}: The sample is from $\sB_n$, but the tester incorrectly classifies it as from $\sA_n$.
\end{itemize}

In general, we take the null hypothesis as the more challenging case and aim to control the type-I error rate within a constant threshold. Under this constraint, the goal is to minimize the type-II error as much as possible. To distinguish between two sets $\sA_n$ and $\sB_n$, we use a quantum measurement $\{M_n, I-M_n\}$. Since we aim to control the discrimination errors for any state within the given sets, regardless of which one is drawn, the type-I error is defined by
\begin{align}
    \alpha(\sA_n, M_n): = \sup_{\rho_n \in \sA_n} \tr[\rho_n (I-M_n)],
\end{align}
and the type-II error is defined by
\begin{align}
    \beta(\sB_n, M_n): = \sup_{\sigma_n \in \sB_n} \tr[\sigma_n M_n].
\end{align}
The asymmetric setting then seeks to determine the optimal exponent at which the type-II error probability decays, known as the \emph{Stein's exponent}, while keeping the type-I error within a fixed threshold $\ve$. Specifically, the goal is to evaluate:
\begin{align}\label{eq: operational quantity in discrimination}
\beta_\ve(\sA_n\|\sB_n) := \inf_{0\leq M_n\leq I}\left\{\beta(\sB_n, M_n): \alpha(\sA_n, M_n) \leq \ve\right\}.
\end{align}

\subsection{{Quantum Stein's lemma with composite hypotheses}}

Before presenting our {quantum Stein's lemma with composite hypotheses}, we first show that the operational quantity $\beta_\ve(\sA_n\|\sB_n)$ in Eq.~\eqref{eq: operational quantity in discrimination} is given by the hypothesis testing relative entropy between two sets of quantum states $\sA_n$ and $\sB_n$, as mathematically defined in Definition~\ref{def: divergence between two sets}.

\begin{shaded}
\begin{lemma}\label{lem: hypothesis operational and DH}
Let $\sA \subseteq \density$ and $\sB \subseteq \PSD$ be two convex sets. For any $\ve \in (0,1)$, it holds that 
\begin{align}
-\log \beta_\ve(\sA\|\sB) = D_{\Hypo,\ve}(\sA\|\sB).
\end{align}
\end{lemma}
\end{shaded}
\begin{proof}
It is worth noting that, unless both $\sA$ and $\sB$ are singletons, the result does not follow trivially from their definitions. This is because the left-hand side represents an operational quantity, while the right-hand side is a mathematical generalization. To establish this result, it suffices to demonstrate that
\begin{align}
    \beta_\ve(\sA\|\sB) = \inf_{0 \leq M \leq I} \left\{\sup_{\sigma \in \sB} \tr[\sigma M]: \sup_{\rho \in \sA} \tr[\rho (I-M)] \leq \ve\right\} =  \sup_{\rho\in\sA} \sup_{\sigma \in \sB} \beta_\ve(\rho\|\sigma).
\end{align}
This requires pulling the supremum over $\sA$ in the condition and the supremum over $\sB$ in the objective function to the left-hand side of $\inf_{0 \leq M \leq I}$, which is a non-trivial task.
First observe that we can write:
\begin{align}
\beta_{\ve}(\sA,\sB) &= 
\inf_{0 \leq M \leq I} \left\{ \sup_{\sigma \in \sB} \tr[\sigma M] : \sup_{\rho \in \sA} \tr[\rho ((1-\ve)I - M)] \leq 0 \right\}\\
&= \inf_{0 \leq M \leq I} \sup_{\sigma \in \sB} \sup_{\rho \in \sA} \sup_{z > 0} \tr[\sigma M] + z \tr[\rho ((1-\ve)-M)]
\end{align}
since $\sup_{z > 0} \tr[\sigma M] + z\tr[\rho((1-\ve)I-M)] = \tr[\sigma M]$ if $\tr[\rho((1-\ve)I-M)] \leq 0$ and $+\infty$ otherwise.
By doing the change of variables $\rho'=z\rho \in \Cone(\sA)$ where $\Cone(\sA)$ is the convex cone generated by $\sA$, we can write:
\begin{align}
\beta_{\ve}(\sA,\sB) = \inf_{0 \leq M \leq I}\; \sup_{\substack{\rho' \in \Cone(\sA)\\ \sigma \in \sB}} \tr[\sigma M] + \tr[\rho' ((1-\ve)I-M)]
\end{align}
Define
\begin{align}
f(M,\rho',\sigma) = \tr[\sigma M] + \tr[\rho' ((1-\ve)I-M)]
\end{align}
which is linear in $M$ for $(\rho',\sigma)$ fixed, and vice-versa. Furthermore the set $\{M : 0\leq M \leq I\}$ is convex compact, and $\Cone(\sA) \times \sB$ is convex. So Sion's minimax theorem~\cite[Corollary 3.3]{Sion1958} applies and we can write
\begin{align}
\beta_{\ve}(\sA,\sB) &= \sup_{\substack{\rho' \in \Cone(\sA)\\ \sigma \in \sB}}\; \inf_{0 \leq M \leq I} \tr[\sigma M] + \tr[\rho' ((1-\ve)I-M)]\\
&= \sup_{\substack{\rho \in \sA\\ \sigma \in \sB}} \; \sup_{z > 0} \;\inf_{0 \leq M \leq I} \tr[\sigma M] + z \tr[\rho ((1-\ve) - M)]\\
&= \sup_{\substack{\rho \in \sA\\ \sigma \in \sB}} \; \inf_{0 \leq M \leq I}\; \sup_{z > 0} \tr[\sigma M] + z \tr[\rho ((1-\ve)I - M)]\\
&= \sup_{\substack{\rho \in \sA\\ \sigma \in \sB}} \inf_{0\leq M\leq I} \{ \tr[\sigma M] : \tr[\rho((1-\ve)I-M)] \leq 0 \}\\
&= \sup_{\substack{\rho \in \sA\\ \sigma \in \sB}}\; \beta_{\ve}(\rho,\sigma)
\end{align}
where in the third line we applied Sion's minimax theorem again to interchange the $\sup_{z > 0}$ with the $\inf_{0\leq M\leq 1}$. This completes the proof. An alternative proof is given in Appendix~\ref{sec: appendix Alternative proof lem: hypothesis operational and DH}.
\end{proof}

Due to the above relation between the optimal type-II error $\beta_\ve(\sA\|\sB)$ and the hypothesis testing relative entropy $D_{\Hypo,\ve}(\sA\|\sB)$, we can now apply the generalized AEP in Theorem~\ref{thm: generalized AEP} to derive a quantum Stein's lemma as follows. 

\begin{shaded}
\begin{theorem}({Quantum Stein's lemma with composite hypotheses.})\label{thm: generalized Steins}
Let $\{\sA_n\}_{n\in\NN}$ and $\{\sB_n\}_{n\in\NN}$ be two sequences of sets satisfying Assumption~\ref{ass: steins lemma assumptions} and $\sA_n \subseteq \density(\cH^{\ox n})$, $\sB_n \subseteq \PSD(\cH^{\ox n})$ and $D_{\max}(\sA_n\|\sB_n) \leq cn$, for all $n \in \NN$ and a constant $c \in \RR_{\pl}$.
Then for any $\ve \in (0,1)$, it holds that
\begin{align}\label{eq: adversarial quantum Steins lemma}
    \lim_{n\to \infty} - \frac{1}{n} \log \beta_\ve(\sA_n\|\sB_n) = D^\reg(\sA\|\sB).
\end{align}
\end{theorem}
\end{shaded}
\begin{proof}
This is a direct consequence of Theorem~\ref{thm: generalized AEP} and Lemma~\ref{lem: hypothesis operational and DH}.
\end{proof}

Note that Theorem~\ref{thm: generalized Steins} is incomparable to the generalized quantum Stein lemma of~\cite{Brand_o_2010,hayashi2024generalized,lami2024solutiongeneralisedquantumsteins}. It is weaker in the sense that we have the additional assumption on the stability of the polar sets under tensor product (see (A.4) on Assumption~\ref{ass: steins lemma assumptions}) for the alternative hypothesis $\sB_n$, but it is stronger in that it allows for {more variants of} composite null hypothesis and in addition we can obtain efficient and controlled approximations of the Stein exponent as presented in Theorem~\ref{thm: generalized AEP}.

\section{Second law of resource manipulation with partial information}
\label{sec: Quantum resource theory with partial information and its reversibility}

In this section, we introduce a new framework of quantum resource theory with partial information and prove that such a theory is reversible under asymptotically resource non-generating operations, identifying the regularized quantum relative entropy between two sets as the ``unique'' measure of resource in this new framework.

\subsection{Operational setting: resource manipulation with partial information}

Quantum resource theory is a framework that finds great success in studying different quantum resource features in recent years (see e.g.~\cite{ChitambarGour19} for an introduction). A standard resource theory is built upon a set of \emph{free operations} $\Omega$ and a set of \emph{free states} $\sF$ (in contrast to resource states).  Take the entanglement theory as an example, $\sF$ consists of the separable (unentangled) states, and local operations and classical communication (LOCC) is a standard choice of $\Omega$. In principle, $\sF$ and $\Omega$ can be adaptively defined, which give rise to a wide variety of meaningful resource theories, as long as they follow a \emph{golden rule}: any free operation can only map a free state to another free state, i.e.~$\Lambda(\rho)\in\sF, \forall\rho\in\sF, \forall\Lambda \in\Omega$. Whether some resource state can be (approximately) converted to another by certain free operations is a fundamental type of problem in quantum information~\cite{fang2020no,fang2022no,theurer2025single}.

Existing studies on quantum resource theory require precise characterization of the source and the target states, and impose an i.i.d. structure, that is, considering transformation $\rho^{\ox n}$ to $\sigma^{\ox m}$ and evaluating the rate of transformation $m/n$. A standard approach to obtain a precise characterization of the state is quantum tomography, which is a resource-intensive task. In practical scenarios, full knowledge of the source state may be unavailable, and different copies of the sources can exhibit correlation by nature. For instance, when establishing entangled photons over a remote distance in a quantum network~\cite{fang2023quantum}, the environmental noise affecting the states can vary over time for different copies of the entangled state, making the precise description of the shared state unclear.

Motivated by this, we propose a new framework of quantum resource theory in this section that works with resource manipulation with partial information. A particular case in quantum thermodynamics has been recently studied in~\cite{watanabe2024black}. More specifically, we consider the transformation of an unknown quantum state $\rho_n \in \density(A^{\ox n})$ to another quantum state $\sigma_m \in \density(B^{\ox m})$. In this case, we do not have complete information about the specific states $\rho_n$ and $\sigma_m$. All we know is that the source state lies within a certain set of quantum states $\sA_n$ (for example, we may know the range of noise affecting the source). Our goal is to design a resource manipulation protocol $\Lambda$ such that regardless of the specific state in $\sA_n$, it \emph{universally} transforms this state into a target state that falls within the target range $\sB_n$. A comparison of this setting with the standard resource manipulation is depicted in Figure~\ref{fig:QRTPI}.

\begin{figure}[H]
    \centering
    \includegraphics[width=14cm]{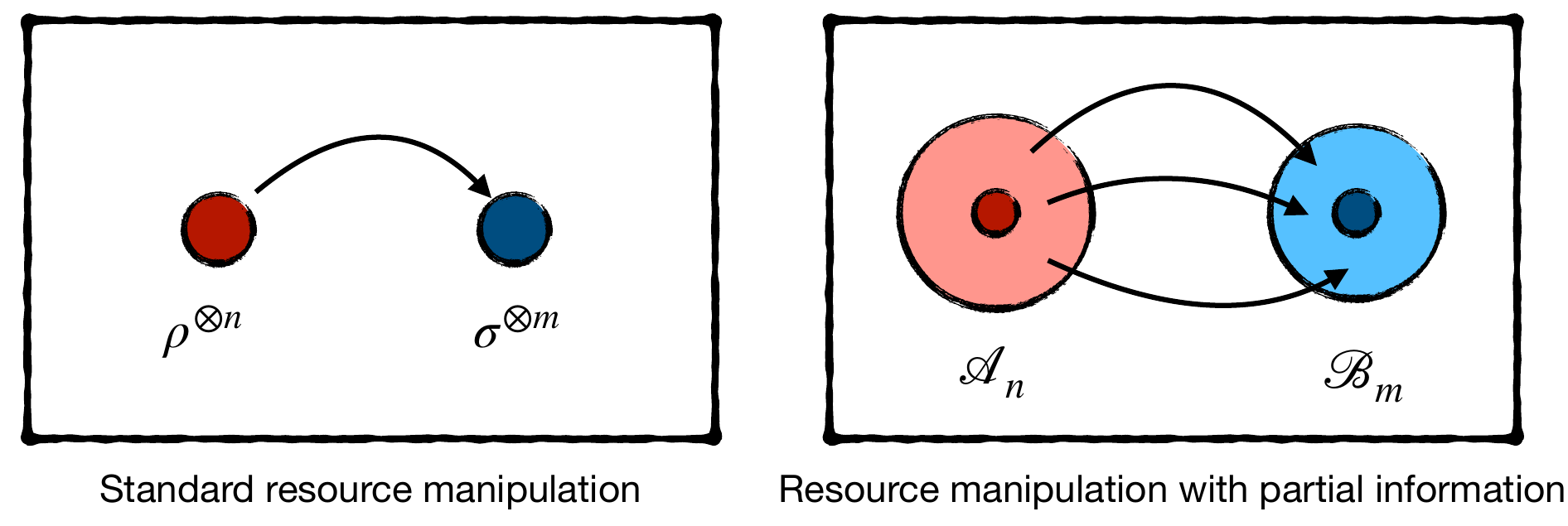}
    \caption{\small Illustration of the standard quantum resource manipulation (left) that transforms an i.i.d. state $\rho^{\ox n}$ to another i.i.d. state $\sigma^{\ox m}$ and the quantum resource manipulation with partial information (right) that transforms an unknown quantum state in $\sA_n$ to a quantum state in the target range $\sB_m$.}
    \label{fig:QRTPI}
\end{figure}

\subsection{Reversibility of resource manipulation}

Let $\Delta(\rho,\sigma):= \frac{1}{2}\|\rho-\sigma\|_1$ be the trace distance between two quantum states. Consider transformation from {a ``black-box resource''} $\sA_n$ to {another ``black-box resource''} $\sB_m$ via free operation class $\Omega_n$, and define the transformation error by 
\begin{align}
    \Delta\left(\sA_n\xrightarrow[]{\Omega_n}\sB_m\right) := \inf_{\Lambda_n \in \Omega_n}\sup_{\rho_n \in \sA_n} \inf_{\sigma_m\in \sB_m}  \Delta(\Lambda_n(\rho_n),\sigma_m).
\end{align}
That is, there exists a free operation $\Lambda_n \in \Omega_n$ such that it universally transforms any state in $\rho_n \in \sA_n$ to a quantum state $\sigma_m \in \sB_m$ within the error $ \Delta\left(\sA_n\xrightarrow[]{\Omega_n}\sB_m\right)$. {Note that the output state $\Lambda_n(\rho_n)$ may vary depending on the input.} Based on this notion of error measure, we can introduce the asymptotic rate of transformation by 
\begin{align}
    r\left(\sA\xrightarrow[]{\Omega} \sB\right):= \sup \left\{\limsup_{n\to \infty} \frac{m_n}{n}: \lim_{n\to \infty} \Delta\left(\sA_n\xrightarrow[]{\Omega_n}\sB_{m_n}\right) = 0\right\},
\end{align}
which is the optimal achievable rate of transformation with vanishing error.

As the previous generalized quantum Stein's lemma~\cite{Brand_o_2010,hayashi2024generalized,lami2024solutiongeneralisedquantumsteins} is deeply related to the reversibility of quantum resource manipulation, here we show an analogy result in this new framework. We say a quantum resource theory with partial information is \emph{reversible} if the asymptotic rate satisfies,
\begin{align}
r\left(\sA\xrightarrow[]{\Omega} \sB\right) \cdot r\left(\sB\xrightarrow[]{\Omega} \sA\right) = 1,
\end{align}
which means we can transform {(or simulate)} from $\sA$ to $\sB$ and then back from $\sB$ to $\sA$ without information loss.
For this, we consider the asymptotically resource non-generating operations $\RNG$ introduced in~\cite{Brand_o_2010}. A quantum operation $\Lambda$ is called $\delta$-resource nongenerating operation  if for every free state $\omega \in \sF$, we have $\cR(\Lambda(\omega)) \leq \delta$, where 
\begin{align}
    \cR(\rho) := \min_{\pi \in \density} \left\{s \geq 0: \frac{\rho + s\pi}{1+s} \in \sF\right\}
\end{align}
is the global robustness of $\rho$. Simple calculation tells that, $\log (\cR(\rho) + 1) = D_{\max}(\rho\|\sF)$. Then we denote the set of all $\delta$-resource nongenerating operations by $\RNG(\delta)$. An asymptotically resource nongenerating operation is a sequence of operations $\{\Lambda_n\}_{n\in \NN}$ such that $\Lambda_n \in \RNG(\delta_n)$ and $\lim_{n\to \infty}\delta_n = 0$. With this, the asymptotic rate of transformation under RNG operations is given by
\begin{align}
    r\left(\sA\xrightarrow[]{\RNG} \sB\right) = \sup \left\{\limsup_{n\to \infty} \frac{m_n}{n}: \lim_{n\to \infty} \Delta\left(\sA_n\xrightarrow[]{\RNG(\delta_n)}\sB_{m_n}\right) = 0, \lim_{n\to \infty}\delta_n = 0\right\}.
\end{align}

The following result gives a precise characterization of this asymptotic rate, which identifies the regularized quantum relative entropy between two sets as the ``unique'' measure of a resource in the asymptotic limit. It also implies the reversibility (also referred to as the second law) of the resource manipulation under RNG.

\begin{shaded}
\begin{theorem}
Let $\{\sA_n\}_{n\in\NN}$, $\{\sB_n\}_{n\in\NN}$ and  $\{\sF_n\}_{n\in\NN}$ be threes sequences of sets satisfying Assumption~\ref{ass: steins lemma assumptions} where $\sA_n \subseteq \density(\cH^{\ox n})$ and $\sB_n \subseteq \density(\cH^{\ox n})$ are the sets of source states and the sets of target states, respectively, and $\sF_n \subseteq \density(\cH^{\ox n})$ are the sets of free states, which defines the operation class $\RNG$. Let $\kappa := \limsup_{n\to \infty} \frac{1}{n} \max_{\omega \in \density} D(\omega\|\sF_n) < \infty$ and $D^\reg(\sB\|\sF) > 0$. Then it holds that
\begin{align}
    r\left(\sA\xrightarrow[]{\RNG} \sB\right) = \frac{D^\reg(\sA\|\sF)}{D^\reg(\sB\|\sF)}.
\end{align}
As a consequence, if $D^\reg(\sA\|\sF) > 0$, it holds that
\begin{align}
    r\left(\sA\xrightarrow[]{\RNG} \sB\right)r\left(\sB\xrightarrow[]{\RNG} \sA\right) = 1.
\end{align}
\end{theorem}
\end{shaded}

\begin{proof}
The proof applies the generalized AEP in Theorem~\ref{thm: generalized AEP} and follows a similar argument in~\cite{brandao2015reversible} and also a more detailed version in~\cite[Chapter 11]{gour2024resources}. This proof contains the converse and the achievable parts, and we will prove the converse part first.

\textit{1) Proof of the converse part:} Suppose by contradiction that
\begin{align}
    r\left(\sA\xrightarrow[]{\RNG} \sB\right) > \frac{D^\reg(\sA\|\sF)}{D^\reg(\sB\|\sF)} + 2\delta,
\end{align} 
for some small positive $\delta$. By definition, this means in particular that for sufficiently small $\ve \in (0,1)$, there exists an achievable rate $r$ such that 
\begin{align}
    r > \frac{D^\reg(\sA\|\sF)}{D^\reg(\sB\|\sF)} + \delta.
\end{align}
That is, there exists a sequence $\{m_n\}_{n\in \NN}\subseteq \NN$ such that $r = \lim_{n\to \infty} \frac{m_n}{n}$ and there exists another sequence $\{\delta_n\}_{n \in \NN} \subseteq \RR_{\pl}$ with a limit of zero such that for every $n \in \NN$, 
\begin{align}
   \Delta\left(\sA_n\xrightarrow[]{\RNG(\delta_n)}\sB_{m_n}\right) \leq \ve.
\end{align}
This means that there exists $\Lambda_n \in \RNG(\delta_n)$ such that for any $\rho_n \in \sA_n$  
\begin{align}
 \Delta(\Lambda_n(\rho_n),\sigma_{m_n}) \leq \ve,
\end{align}
for some  $\sigma_{m_n} \in \sB_{m_n}$.
By using the continuity in~\cite[Theorem 10.2.2]{gour2024resources}, we have 
\begin{align}
    |D(\Lambda_n(\rho_n)\|\sF_{m_n}) - D(\sigma_{m_n}\|\sF_{m_n})| \leq c_{m_n} \ve + (1+\ve) h\left(\frac{\ve}{1+\ve}\right),
\end{align}
where $c_{m_n}:= \max_{\omega \in \density} D(\omega\|\sF_{m_n})$ and $h(x)$ is the binary Shannon entropy. This gives
\begin{align}
    D(\sigma_{m_n}\|\sF_{m_n}) \leq D(\Lambda_n(\rho_n)\|\sF_{m_n}) + c_{m_n} \ve + (1+\ve) h\left(\frac{\ve}{1+\ve}\right).
\end{align}
Since $D(\sB_{m_n}\|\sF_{m_n}) \leq D(\sigma_{m_n}\|\sF_{m_n})$ by definition, we get
\begin{align}\label{eq: GRT converse proof tmp1}
    D(\sB_{m_n}\|\sF_{m_n}) \leq D(\Lambda_n(\rho_n)\|\sF_{m_n}) + c_{m_n}  \ve + (1+\ve) h\left(\frac{\ve}{1+\ve}\right).
\end{align}
Now, let $\omega_{n} \in \sF_{n}$  be the optimizer satisfying $D(\rho_{n}\|\omega_{n}) = D(\rho_{n}\|\sF_{n})$.
Then we have 
\begin{align}
    D(\Lambda_n(\rho_n)\|\sF_{m_n}) & = \min_{\tau_{m_n} \in \sF_{m_n}} D(\Lambda_n(\rho_n)\|\tau_{m_n}) \\
    & \leq D(\Lambda_n(\rho_n)\|\Lambda_n(\omega_n)) + \min_{\tau_{m_n} \in \sF_{m_n}} D_{\max}(\Lambda_n(\omega_n)\|\tau_{m_n})\\
    & \leq D(\rho_n\|\omega_{n}) + D_{\max}(\Lambda_n(\omega_n)\|\sF_{m_n})
\end{align}
where the first inequality follows from the triangle inequality $D(\rho\|\sigma) \leq D(\rho\|\omega) + D_{\max}(\omega\|\sigma)$ (see e.g.~\cite[Theorem 6.3.3]{gour2024resources}), and the second inequality follows from the data-processing inequality of quantum relative entropy. Since $\omega_n$ is a free state and $\Lambda_n \in \RNG(\delta_n)$, the global robustness of $\Lambda_n(\omega_n)$ cannot exceed $\delta_n$ by definition and in particular $D_{\max}(\Lambda_n(\omega_n)\|\sF_{m_n}) \leq \log(1+\delta_n)$. Therefore,
\begin{align}
    D(\Lambda_n(\rho_n)\|\sF_{m_n}) \leq D(\rho_n\|\sF_n) + \log (1+\delta_n),
\end{align}
where we use the optimality assumption of $\omega_n$ here.
Taking this into Eq.~\eqref{eq: GRT converse proof tmp1}, we get
\begin{align}
     D(\sB_{m_n}\|\sF_{m_n}) \leq D(\rho_n\|\sF_n) + \log (1+\delta_n) + c_{m_n} \ve + (1+\ve) h\left(\frac{\ve}{1+\ve}\right).
\end{align}
As this holds for any $\rho_n \in \sA_n$, we have 
\begin{align}
    D(\sB_{m_n}\|\sF_{m_n}) \leq D(\sA_n\|\sF_n) + \log (1+\delta_n) + c_{m_n}  \ve + (1+\ve) h\left(\frac{\ve}{1+\ve}\right).
\end{align}
Noting that $\lim_{n\to \infty} m_n/n = r$ and $\lim_{n\to \infty} \frac{1}{m_n}c_{m_n} \leq \kappa$, we have 
\begin{align}
    \lim_{n\to \infty} \frac{1}{m_n} D(\sB_{m_n}\|\sF_{m_n})  \leq \lim_{n\to \infty} \frac{n}{m_n} \frac{1}{n} D(\sA_n\|\sF_n) + \kappa \ve.
\end{align}
That is,
\begin{align}
    D^\reg(\sB\|\sF) \leq \frac{1}{r} D^\reg(\sA\|\sF) + \kappa\ve.
\end{align}
However, since $r > \frac{D^\reg(\sA\|\sF)}{D^\reg(\sB\|\sF)} + \delta$ for sufficiently small $\ve \in (0,1)$, we get a contradiction that 
\begin{align}
    D^\reg(\sB\|\sF) \leq \frac{1}{r} D^\reg(\sA\|\sF) + \kappa\ve < D^\reg(\sB\|\sF),
\end{align}
by taking a small enough $\ve \in (-\infty, \delta D^\reg(\sB\|\sF)/\kappa) \cap (0,1)$. This concludes the converse part.

\textit{2) Proof of the achievable part:} For any
\begin{align}
    r < \frac{D^\reg(\sA\|\sF)}{D^\reg(\sB\|\sF)},
\end{align}
we aim to show that $r$ is an achievable rate. For this, we fix $\ve \in (0,1)$ and denote $m_n := \lceil nr\rceil$ so that $\lim_{n\to \infty} m_n/n = r$.
Let $M_n$ be the optimal test for $\beta_\ve(\sA_n\|\sF_n)$. That is, for any $\delta > 0$ (will be determined later) and sufficiently large $n$, we have from the generalized AEP in Theorem~\ref{thm: generalized AEP},
\begin{align}\label{eq: reversibility acheivable proof tmp1}
    \sup_{\rho_n \in \sA_n} \tr[\rho_n (I-M_n)] \leq \frac{\ve}{2}, \quad \text{and}\quad \sup_{\omega_n \in \sF_n} \tr[\omega_n M_n] \leq 2^{-n(D^\reg(\sA\|\sF) - \delta)}.
\end{align}
Let $\sigma_{m_n}^*$ and $\sigma_{m_n}'$ be the optimizers satisfying 
\begin{align}\label{eq: reversibility acheivable proof tmp2}
    D_{\max,\ve/2}(\sB_{m_n}\|\sF_{m_n}) = D_{\max,\ve/2}(\sigma_{m_n}^*\|\sF_{m_n}) = D_{\max}(\sigma_{m_n}'\|\sF_{m_n})
\end{align}
where $\Delta(\sigma_m',\sigma_m^*) \leq \ve/2$~\footnote{Here we use trace distance for the definition of the smoothed max-relative entropy for convenience.}. 
Define the sequence of maps 
\begin{align}
    \Lambda_n(X):= \tr[M_n X] \sigma_{m_n}' + \tr[(I-M_n)X] \omega_{m_n},
\end{align}
where $\omega_{m_n} \in \sF_{m_n}$ is a free state. Then we can check that these operations give the expected transformation within the error threshold, and they are indeed asymptotically resource nongenerating operations. For this, we can first check that for any $\rho_n \in \sA_n$,
\begin{align}
    \Delta(\Lambda_n(\rho_n),\sigma_{m_n}^*) & \leq \Delta(\Lambda_n(\rho_n),\sigma_{m_n}') + \Delta(\sigma_{m_n}',\sigma_{m_n}^*)\\
    & = (1-\tr[M_n\rho_n]) \Delta(\omega_{m_n}, \sigma_{m_n}') + \Delta(\sigma_{m_n}',\sigma_{m_n}^*)\\
    & \leq \ve,
\end{align}
where the first inequality follows from the triangle inequality of the trace distance, the equality follows by evaluating $\Lambda_n(\rho_n)$ using the definition and the second inequality follows by Eq.~\eqref{eq: reversibility acheivable proof tmp1} and the fact that $\Delta(\omega_{m_n}, \sigma_{m_n}') \leq 1$.
As this holds for any $\rho_n \in \sA_n$, we have
\begin{align}
  \sup_{\rho_n \in \sA_n} \inf_{\sigma_{m_n}\in \sB_{m_n}}  \Delta(\Lambda_n(\rho_n),\sigma_{m_n}) \leq \ve.
\end{align}
Then we check $\Lambda_n \in \RNG(\delta_n)$ with $\delta_n = 2^{-n\delta}$ and therefore $\lim_{n\to \infty} \delta_n = 0$ (note that $\delta$ will be determined later). We have the following
\begin{align}
D^\reg(\sB\|\sF) & = \lim_{n\to \infty} \frac{1}{m_n} D_{\max,\ve/2}(\sB_{m_n}\|\sF_{m_n})\\
&= \lim_{n\to \infty} \frac{1}{m_n} D_{\max}(\sigma_{m_n}'\|\sF_{m_n}) \\
&= \lim_{n\to \infty} \frac{n}{m_n} \frac{1}{n} D_{\max}(\sigma_{m_n}'\|\sF_{m_n}) \\
&= \frac{1}{r} \lim_{n\to \infty} \frac{1}{n} D_{\max}(\sigma_{m_n}'\|\sF_{m_n}),\label{eq: reversibility acheivable proof tmp3}
\end{align}
where the first line follows from the generalized AEP in Theorem~\ref{thm: generalized AEP}, the second line follows by the optimality assumption of $\sigma_{m_n}'$ in Eq.~\eqref{eq: reversibility acheivable proof tmp2} and the last equality follows since $\lim_{n\to \infty} m_n / n = r$.
Since $r < \frac{D^\reg(\sA\|\sF)}{D^\reg(\sB\|\sF)}$, there exists $\delta > 0$ such  that $r < (D^\reg(\sA\|\sF) - 2\delta)/D^\reg(\sB\|\sF)$, or equivalently,
\begin{align}
    r D^\reg(\sB\|\sF) < D^\reg(\sA\|\sF) - 2\delta.
\end{align}
Taking this into Eq.~\eqref{eq: reversibility acheivable proof tmp3}, we have for sufficiently large $n$,
\begin{align}
    D_{\max}(\sigma_{m_n}'\|\sF_{m_n}) \leq n (D^\reg(\sA\|\sF) - 2\delta).
\end{align}
That is, the global robustness of $\sigma_{m_n}'$ satisfies 
\begin{align}\label{eq: reversibility acheivable proof tmp4}
    \cR(\sigma_{m_n}') \leq 2^{n (D^\reg(\sA\|\sF) - 2\delta)} - 1.
\end{align}
To show that $\Lambda_n \in \RNG(\delta_n)$, let $\eta_n \in \sF_n$ and denote by $t_n := \tr[M_n \eta_n]$ and $r_{m_n} := \cR(\sigma_{m_n}')$. From the convexity of global robustness we get
\begin{align}
    \cR(\Lambda_n(\eta_n)) & \leq t_n R(\sigma_{m_n}') + (1-t_n) R(\omega_{m_n}) \\
    & = t_n r_{m_n}\\
    & \leq t_n(r_{m_n} +1) \\
    & \leq t_n 2^{n (D^\reg(\sA\|\sF) - 2\delta)}\\
    & \leq 2^{-n\delta}\\
    & = \delta_n
\end{align}
where the second line follows because $\omega_{m_n} \in \sF_{m_n}$ ant therefore $R(\omega_{m_n}) = 0$, the fourth line follows by Eq.~\eqref{eq: reversibility acheivable proof tmp4} and the fifth line follows as $t_n \leq 2^{-n(D^\reg(\sA\|\sF) - \delta)}$ by Eq.~\eqref{eq: reversibility acheivable proof tmp1}. This concludes that $\Lambda_n \in \RNG(\delta_n)$ and therefore concludes the proof of the achievable part.
\end{proof}

\section{Discussion}
\label{sec: conclusion}

The significance of this work is multifaceted. First, we established a very general quantum AEP with explicit convergence guarantees and computational efficiency. Given that the AEP is central to information theory and is evident in its broad applications in both classical and quantum information theory, we expect that our generalized quantum AEP plays a similar role and advances the study of quantum information science further. Second, our {quantum Stein's lemma with composite hypotheses}, {combined with the accompanying work~\cite{fang2025efficient},} addresses open questions raised by Brand\~{a}o et al.~\cite[Section 3.5]{brandao2020adversarial} and Mosonyi et al.~\cite[Section VI]{Mosonyi_2022}, which seek a Stein's lemma with computational efficiency. We also proposed a new framework for quantum resource theory in which state transformations are performed without requiring precise characterization of the states being manipulated, making it more robust to imperfections.
Finally, we believe that the technical tools established in this work—including the variational formula, superadditivity for measured relative entropy between two sets of states—are of independent interest and are likely to have further applications.

Many problems remain open for future investigation. While we have obtained explicit bounds on the second-order term for the generalized AEP of order $O(n^{2/3} \log n)$, such bounds are not tight in general. Can one bound the second-order term by $O(\sqrt{n})$ as in the AEP? Sharpening the second-order term would provide a more comprehensive understanding of the convergence behavior and enhance its accuracy in practical applications. 
We proved the asymptotic reversibility of quantum resource theory under asymptotic resource non-generating operations. But the systematic study of other different operational classes and settings are left for future exploration. 

\vspace{1cm}

\noindent \textbf{Acknowledgements.} 
O.F. would like to thank Roberto Rubboli and Marco Tomamichel for a discussion on the additivity properties of the measured relative entropy at the Centre for Quantum Technologies in Singapore. We would also like to thank David Sutter for a question about the relation to~\cite{Mazzola2025} which motivated us to add analogous superadditivity statements (e.g., Lemma~\ref{lem: generalized supadditivity DM alpha}) for $\alpha > 1$. 
{We thank anonymous reviewers for suggesting a more elementary proof of $D_{\Meas,1/2}(\rho\|\sigma) \geq D_{\min}(\rho\|\sigma)$ in Lemma~\ref{thm: comparison of quantum divergence}, as well as the discussion in Appendix~\ref{app: partial trace assumption}.}

K.F. is supported in part by the National Natural Science Foundation of China (Grant No. 92470113 and 12404569), the Shenzhen Science and Technology Program (Grant No. QNXMB20 250701091826036 and JCYJ20240813113519025), the Shenzhen Fundamental Research Program (Grant No. JCYJ20241202124023031), the General R\&D Projects of 1+1+1 CUHK-CUHK (SZ)-GDST Joint Collaboration Fund (Grant No. GRDP2025-022), and the University Development Fund (Grant No. UDF01003565). O.F. acknowledges support by the European Research Council (ERC Grant AlgoQIP, Agreement No. 851716), by the European Union’s Horizon research and innovation programme under the project VERIqTAS (Grant Agreement No 101017733).

\paragraph{Statement.} There is no conflict of interest and no data generated in this work.

\bibliographystyle{alpha_abbrv}
\bibliography{Bib}

\newpage
\vspace{2cm}
\appendix

\noindent \textbf{\LARGE Appendices}

\vspace{1cm}

In Section~\ref{app sec: variational formula for other alpha} we present the variational formula for the measured relative entropy with different parameters $\alpha$. In Section~\ref{app: partial trace assumption}, we discuss the connection of our assumptions with the partial trace assumption. In Section~\ref{sec: appendix Alternative proof lem: hypothesis operational and DH} we provide an alternative proof for Lemma~\ref{lem: hypothesis operational and DH}.  
Some useful properties are presented in Section~\ref{sec: app Useful properties}.

\section{Variational formula for $D_{\Meas,\alpha}(\sA\|\sB)$}
\label{app sec: variational formula for other alpha}

\begin{proof}[Proof of Lemma~\ref{lem: DM alpha variational sets}]
The case for $\alpha \in [1/2,1)$ has been proved in the main text. Now we prove the other two cases.
Let $\alpha \in (0,1/2)$. For any fixed $\rho \in \sA$ and $\sigma \in \sB$, we have by Eq.~\eqref{eq: DM alpha variational} that
\begin{align}
Q_{\Meas,\alpha}(\rho\|\sigma) & := \inf_{W \in \PD}\ \alpha \tr \left[\rho W\right] + (1-\alpha) \tr [\sigma W^{\frac{\alpha}{\alpha-1}}]\\
& = \inf_{\substack{W \in \PD\\W^{\frac{\alpha}{\alpha-1}} \leq V}}\ \alpha \tr \left[\rho W\right] + (1-\alpha) \tr [\sigma V],
\end{align}
where the second equality follows by introducing an additional variable $V$.
Then we have
\begin{align}
\sup_{\substack{\rho \in \sA\\ \sigma \in \sB}} Q_\alpha^\Meas(\rho\|\sigma)
& = \sup_{\substack{\rho \in \sA\\ \sigma \in \sB}}  \; \inf_{\substack{W \in \PD\\W^{\frac{\alpha}{\alpha-1}} \leq V}}\ \alpha \tr \left[\rho W\right] + (1-\alpha) \tr [\sigma V].
\end{align}
Note that all $\sA,\sB$ and $\{(W,V): W \in \PD, W^{\frac{\alpha}{\alpha-1}} \leq V\}$ are convex sets, with $\sA,\sB$ being compact. Moreover, the objective function is linear in $(\rho,\sigma)$, and also linear in $(W,V)$. So we can apply Sion's minimax theorem~\cite[Corollary 3.3]{Sion1958} to exchange the infimum and supremum and get
\begin{align}
\sup_{\substack{\rho \in \sA\\ \sigma \in \sB}} Q_\alpha^\Meas(\rho\|\sigma) & {=} \inf_{\substack{W \in \PD\\W^{\frac{\alpha}{\alpha-1}} \leq V}} \; \sup_{\substack{\rho \in \sA\\ \sigma \in \sB}} \ \alpha \tr \left[\rho W\right] + (1-\alpha) \tr [\sigma V]\\
& {=} \inf_{\substack{W \in \PD\\W^{\frac{\alpha}{\alpha-1}} \leq V}} \;  \ \alpha \sup_{\rho \in \sA} \tr \left[\rho W\right] + (1-\alpha) \sup_{\sigma \in \sB} \tr [\sigma V].
\end{align}
By using the definition of the support function, we have
\begin{align}
\sup_{\substack{\rho \in \sA\\ \sigma \in \sB}} Q_\alpha^\Meas(\rho\|\sigma) & = \inf_{\substack{W \in \PD\\W^{\frac{\alpha}{\alpha-1}} \leq V}} \alpha h_{\sA}(W) + (1-\alpha) h_{\sB}(V)\\
& = \inf_{\substack{W \in \PD\\W^{\frac{\alpha}{\alpha-1}} \leq V}} h_{\sA}(W)^\alpha  h_{\sB}(V)^{1-\alpha},
\end{align}
where the second line follows from the weighed arithmetic-geometric mean inequality $\alpha x + (1-\alpha)y \leq x^\alpha y^{1-\alpha}$ (with equality if and only if $x = y$) and the fact that $(W,V)$ is a feasible solution implies $(kW, k^{\frac{\alpha}{\alpha-1}} V)$ is also a feasible solution for any $k \geq 0$. Therefore, we can choose $k = (h_{\sA}(W)/h_{\sB}(W))^{\alpha-1}$, which implies $h_{\sA}(kW) = h_{\sB}(k^{\frac{\alpha}{\alpha-1}} V)$ and therefore the equality of the weighed arithmetic-geometric mean is achieved. Similarly, for any feasible solution $(W, V)$ we can always construct a new solution $(W h_{\sB}(V)^{\frac{1-\alpha}{\alpha}},V/h_{\sB}(V))$ achieves the same objective value. This implies 
\begin{align}
    \sup_{\substack{\rho \in \sA\\ \sigma \in \sB}} Q_\alpha^\Meas(\rho\|\sigma)  = \inf_{\substack{W \in \PD\\W^{\frac{\alpha}{\alpha-1}} \leq V\\h_{\sB}(V) = 1}} (h_{\sA}(W))^\alpha = \inf_{\substack{W \in \PD\\W^{\frac{\alpha}{\alpha-1}} \leq V\\h_{\sB}(V) \leq 1}} (h_{\sA}(W))^\alpha,
\end{align}
where the second equality follows by the same reasoning. 
Finally, noting that $h_{\sB}(V) \leq 1$ if and only if $V \in \sB^{\circ}$, we have the asserted result in Eq.~\eqref{eq: DM alpha variational sets}. It is easy to check that the objective function $h_{\sA}(W)$ is convex in $W$ and the feasible set is also a convex set.

Let $\alpha \in (1,+\infty)$. For any fixed $\rho \in \sA$ and $\sigma \in \sB$, we have by Eq.~\eqref{eq: DM alpha variational} that
\begin{align}
Q_{\Meas,\alpha}(\rho\|\sigma) & := \sup_{W \in \PD} \ \alpha \tr\left[\rho W^{\frac{\alpha-1}{\alpha}}\right] + (1-\alpha) \tr[\sigma W]\\
& = \sup_{\substack{W \in \PD\\W^{\frac{\alpha-1}{\alpha}} \geq V \geq 0}}\ \alpha \tr\left[\rho V\right] + (1-\alpha) \tr[\sigma W],
\end{align}
where the second equality follows by introducing an additional variable $V$.
Then we have
\begin{align}
\inf_{\substack{\rho \in \sA\\ \sigma \in \sB}} Q_\alpha^\Meas(\rho\|\sigma)
& = \inf_{\substack{\rho \in \sA\\ \sigma \in \sB}}  \; \sup_{\substack{W \in \PD\\W^{\frac{\alpha-1}{\alpha}} \geq V \geq 0}}\ \alpha \tr\left[\rho V\right] + (1-\alpha) \tr[\sigma W].
\end{align}
Note that all $\sA,\sB$ and $\{(W,V): W \in \PD, W^{\frac{\alpha-1}{\alpha}} \geq V \geq 0\}$ are convex sets, with $\sA,\sB$ being compact. Moreover, the objective function is linear in $(\rho,\sigma)$, and also linear in $(W,V)$. So we can apply Sion's minimax theorem~\cite[Corollary 3.3]{Sion1958} to exchange the infimum and supremum and get
\begin{align}
\inf_{\substack{\rho \in \sA\\ \sigma \in \sB}} Q_\alpha^\Meas(\rho\|\sigma) & =  \sup_{\substack{W \in \PD\\W^{\frac{\alpha-1}{\alpha}} \geq V \geq 0}}\; \inf_{\substack{\rho \in \sA\\ \sigma \in \sB}}  \alpha \tr\left[\rho V\right] + (1-\alpha) \tr[\sigma W]\\
& = \sup_{\substack{W \in \PD\\W^{\frac{\alpha-1}{\alpha}} \geq V \geq 0}}  \;  \ \alpha \inf_{\rho \in \sA} \tr \left[\rho V\right] + (1-\alpha) \inf_{\sigma \in \sB} \tr [\sigma W].
\end{align}
Recall that $\revh_{\cvxset}(X):= \inf_{\rho\in \cvxset} \tr[X \rho]$, we have
\begin{align}
\inf_{\substack{\rho \in \sA\\ \sigma \in \sB}} Q_\alpha^\Meas(\rho\|\sigma) & = \sup_{\substack{W \in \PD\\W^{\frac{\alpha-1}{\alpha}} \geq V \geq 0}} \alpha \revh_{\sA}(V) + (1-\alpha) \revh_{\sB}(W)\\
& = \sup_{\substack{W \in \PD\\W^{\frac{\alpha-1}{\alpha}} \geq V \geq 0}} \revh_{\sA}(V)^\alpha \; \revh_{\sB}(W)^{1-\alpha},
\end{align}
where the second line follows from the weighed arithmetic-geometric mean inequality $\alpha x + (1-\alpha)y \leq x^\alpha y^{1-\alpha}$ (with equality if and only if $x = y$) and the fact that $(W,V)$ is a feasible solution implies $(kW, k^{\frac{\alpha-1}{\alpha}} V)$ is also a feasible solution for any $k \geq 0$. Therefore, we can choose $k = (\revh_{\sA}(V)/\revh_{\sB}(W))^{\alpha}$, which implies $\revh_{\sB}(kW) = \revh_{\sA}(k^{\frac{\alpha-1}{\alpha}} V)$ and therefore the equality of the weighed arithmetic-geometric mean is achieved. Similarly, for any feasible solution $(W, V)$ we can always construct a new solution $(W / \revh_{\sB}(W), V \revh_{\sB}(W)^{\frac{1-\alpha}{\alpha}})$ achieves the same objective value. This implies 
\begin{align}
    \inf_{\substack{\rho \in \sA\\ \sigma \in \sB}} Q_\alpha^\Meas(\rho\|\sigma)  = \sup_{\substack{W \in \PD\\W^{\frac{\alpha-1}{\alpha}} \geq V \geq 0\\ \revh_{\sB}(W) = 1}} (\revh_{\sA}(V))^\alpha = \sup_{\substack{W \in \PD\\W^{\frac{\alpha-1}{\alpha}} \geq V \geq 0\\ \revh_{\sB}(W) \geq 1}} (\revh_{\sA}(V))^\alpha,
\end{align}
where the second equality follows by the same reasoning. 
Finally, noting that $\revh_{\sB}(W) \geq 1$ if and only if $W \in \sB^{\star}:=\{X \in \PSD: \tr[X Y] \geq 1, \forall Y \in \sB\}$, we have the asserted result in Eq.~\eqref{eq: DM alpha variational sets}. It is easy to check that the objective function $\revh_{\sA}(V)$ is concave in $V$ and the feasible set is also a convex set.
\end{proof}

\section{On the partial trace assumption}
\label{app: partial trace assumption}

In this section, we show that conditions in Assumption~\ref{ass: steins lemma assumptions} imply that the sets under consideration are actually closed under partial trace. We thank an anonymous reviewer for pointing out this connection and for suggesting the proof presented below.

We first consider the case where all sets consist of normalized states.

\begin{lemma}\label{lem: partial trace normalized states appendix}
Let $\cH_1$ and $\cH_2$ be finite-dimensional Hilbert spaces, and consider compact convex sets of quantum states $\sA_1 \subseteq \density(\cH_1)$, $\sA_2 \subseteq \density(\cH_2)$, and $\sA_{12} \subseteq \density(\cH_1 \ox \cH_2)$. If
\begin{align}\label{eq: partial trace normalized polar appendix}
    \polarPSD{(\sA_1)} \ox \polarPSD{(\sA_2)} \subseteq \polarPSD{(\sA_{12})},
\end{align}
then the sets $\sA$ are closed under partial trace, in the sense that $\rho_{12} \in \sA_{12}$ implies that
\begin{align}
    \rho_1 := \tr_2 \rho_{12} \in \sA_1, \qquad \rho_2 := \tr_1 \rho_{12} \in \sA_2.
\end{align}
\end{lemma}

\begin{proof}
Since $\sA_2$ is composed of states only, we have $I_2 \in \polarPSD{(\sA_2)}$. Given some $R_1 \in \polarPSD{(\sA_1)}$, we thus obtain that $R_1 \ox I_2 \in \polarPSD{(\sA_{12})}$. Hence, for all $\rho_{12} \in \sA_{12}$,
\begin{align}
    1 \geq \tr[\rho_{12}(R_1 \ox I_2)] = \tr[\rho_1 R_1].
\end{align}
Since $R_1 \in \polarPSD{(\sA_1)} = \sA_1^\circ \cap \PSD(\cH_1)$ is arbitrary, we deduce that
\begin{align}
    \rho_1 &\in \left(\polarPSD{(\sA_1)}\right)^\circ \notag\\
    &= \left(\sA_1^\circ \cap \PSD(\cH_1)\right)^\circ \notag\\
    &= \widebar{\operatorname{co}}\left(\sA_1^{\circ\circ} \cup \PSD(\cH_1)^\circ\right) \notag\\
    &= \widebar{\operatorname{co}}\left(\widebar{\operatorname{co}}(\sA_1 \cup \{0\}) \cup (-\PSD(\cH_1))\right) \notag\\
    &= \widebar{\operatorname{co}}\left(\sA_1 \cup (-\PSD(\cH_1))\right) \notag\\
    &= \sA_1 - \PSD(\cH_1) \label{eq: partial trace normalized bipolar appendix}\\
    &= \{X_1: \exists\, \sigma_1 \in \sA_1 \text{ such that } X_1 \leq \sigma_1\}.\notag
\end{align}
Here, the equality in the third line follows from~\cite[Ch.~IV, \S~1.5, Corollary~2]{schaefer1971topological}, while that in the fourth line follows from~\cite[Ch.~IV, \S~1.5, Theorem]{schaefer1971topological}; we also used $\PSD(\cH_1)^\circ=-\PSD(\cH_1)$. In the fifth line we observed that $0 \in -\PSD(\cH_1)$, while the equality in the sixth line can be verified directly using the compactness of $\sA_1$. Since $\rho_1$ is a normalized state, the fact that $\rho_1 \leq \sigma_1$ for some $\sigma_1 \in \sA_1$ implies that $\rho_1=\sigma_1 \in \sA_1$. Exchanging the first and second systems completes the proof.
\end{proof}

The above lemma uses the fact that the $\sA$ sets are composed of normalized states. In general, however, we also consider arbitrary sets of positive semidefinite operators in the second argument. We show that nothing substantial changes in this case, in the sense that closure under partial trace still holds after a trivial modification of the underlying sets. To this end, we start by showing that, without loss of generality, any set $\sB$ that goes into the second
argument of a one-shot divergence can be taken to be closed under the operation of subtracting
positive semidefinite operators (as long as one does not leave the PSD cone).

Given $\sB \subseteq \PSD$, define
\begin{align}\label{eq: appendix B prime definition}
    \sB' := (\sB - \PSD)_{\pl} := (\sB - \PSD) \cap \PSD.
\end{align}

\begin{lemma}\label{lem: appendix B prime divergences}
Given some set $\sB \subseteq \PSD$, define $\sB'$ as in Eq.~\eqref{eq: appendix B prime definition}. Then, for any set of states $\sA$ and any $\ve \in [0,1]$, we have
\begin{align}
    D_{\Hypo,\ve}(\sA\|\sB') & = D_{\Hypo,\ve}(\sA\|\sB),\\
    D_{\max,\ve}(\sA\|\sB') & = D_{\max,\ve}(\sA\|\sB).
\end{align}
\end{lemma}

\begin{proof}
Both divergences are operator anti-monotone in the second argument.
\end{proof}

The above operation turns out to leave the positive polar unchanged, as we now show.

\begin{lemma}\label{lem: appendix B prime polar}
Let $\sB \subseteq \PSD$ be an arbitrary compact convex set of positive semidefinite operators. Then, adopting the definition in Eq.~\eqref{eq: appendix B prime definition}, we have
\begin{align}\label{eq: appendix B prime polar equality}
    \polarPSD{(\sB')} = \left(\polarPSD{(\sB)}\right)' = \polarPSD{(\sB)}.
\end{align}
\end{lemma}

\begin{proof}
Start by noting that
\begin{align}\label{eq: appendix B minus PSD polar}
    (\sB - \PSD)^\circ = \polarPSD{(\sB)}.
\end{align}
Indeed, on the one hand, since $\sB \subseteq \sB - \PSD$, we obtain $(\sB - \PSD)^\circ \subseteq \sB^\circ$. Moreover, $(\sB - \PSD)^\circ$ must contain only positive semidefinite operators; otherwise, the trace against an operator of the form $B-\lambda |\psi\rangle\!\langle\psi|$ would exceed $1$ for any $B\in\sB$, provided that $|\psi\rangle$ is chosen suitably and $\lambda$ is sufficiently large. Hence $(\sB - \PSD)^\circ \subseteq \polarPSD{(\sB)}$. On the other hand, if $R \in \polarPSD{(\sB)}$, then for all $B\in\sB$ and $A\geq0$,
\begin{align}
    \tr[R(B-A)] \leq \tr[RB] \leq 1,
\end{align}
so that $R\in(\sB-\PSD)^\circ$. We deduce that Eq.~\eqref{eq: appendix B minus PSD polar} holds.

Now, we have
\begin{align}
    (\sB')^\circ
    &= \left((\sB-\PSD)\cap\PSD\right)^\circ \notag\\
    &= \widebar{\operatorname{co}}\left((\sB-\PSD)^\circ \cup \PSD^\circ\right) \notag\\
    &= \widebar{\operatorname{co}}\left(\polarPSD{(\sB)} \cup (-\PSD)\right) \notag\\
    &= \polarPSD{(\sB)} - \PSD.
\end{align}
Here, in the second line we used once again~\cite[Ch.~IV, \S~1.5, Corollary~2]{schaefer1971topological}, and in the third line we employed Eq.~\eqref{eq: appendix B minus PSD polar}. Taking the positive part shows that
\begin{align}
    \polarPSD{(\sB')} = \left(\polarPSD{(\sB)}\right)'.
\end{align}
Now, on the one hand, $\sB \subseteq \sB'$, implying that $\polarPSD{(\sB)} \supseteq \polarPSD{(\sB')}$; on the other hand, $\left(\polarPSD{(\sB)}\right)' \supseteq \polarPSD{(\sB)}$. Chaining these two inclusions gives
\begin{align}
    \polarPSD{(\sB)} \subseteq \left(\polarPSD{(\sB)}\right)' = \polarPSD{(\sB')} \subseteq \polarPSD{(\sB)},
\end{align}
which proves Eq.~\eqref{eq: appendix B prime polar equality}.
\end{proof}

\begin{lemma}\label{lem: appendix B prime assumptions}
Given a family $\sB=(\sB_n)_n$ of compact convex sets $\sB_n \subseteq \PSD(\cH^{\ox n})$ that obeys any, or both, of the two assumptions
\begin{align}
    \sB_n \ox \sB_m & \subseteq \sB_{n+m},\label{eq: appendix B tensor assumption}\\
    \polarPSD{(\sB_n)} \ox \polarPSD{(\sB_m)} & \subseteq \polarPSD{(\sB_{n+m})},\label{eq: appendix B polar tensor assumption}
\end{align}
for all $n,m$, the new family $\sB'=(\sB'_n)_n$ defined via Eq.~\eqref{eq: appendix B prime definition} also satisfies the same assumptions.
\end{lemma}

\begin{proof}
Pick some $R_n\in\sB'_n$ and $R_m\in\sB'_m$, so that $0\leq R_n\leq B_n\in\sB_n$ and $0\leq R_m\leq B_m\in\sB_m$. Then
\begin{align}
    0\leq R_n\ox R_m \leq B_n\ox B_m \in \sB_n\ox\sB_m \subseteq \sB_{n+m},
\end{align}
and therefore $R_n\ox R_m\in\sB'_{n+m}$. Since $R_n$ and $R_m$ were completely arbitrary, we deduce that
\begin{align}
    \sB'_n\ox\sB'_m \subseteq (\sB_n\ox\sB_m)' \subseteq \sB'_{n+m}.
\end{align}
The statement about the polar tensor-product assumption follows directly from Lemma~\ref{lem: appendix B prime polar}.
\end{proof}

We are finally ready to show that Assumptions (A.3) and (A.4) together effectively imply closure under partial traces.

\begin{lemma}\label{lem: appendix B prime partial trace}
Let $\cH_1$ and $\cH_2$ be finite-dimensional Hilbert spaces. Consider compact convex sets of positive semidefinite operators $\sB_1 \subseteq \PSD(\cH_1)$, $\sB_2 \subseteq \PSD(\cH_2)$, and $\sB_{12} \subseteq \PSD(\cH_1\ox\cH_2)$. If
\begin{align}\label{eq: appendix B prime partial trace polar assumption}
    \polarPSD{(\sB_1)} \ox \polarPSD{(\sB_2)} \subseteq \polarPSD{(\sB_{12})},
\end{align}
then the associated sets $\sB'$ defined as in Eq.~\eqref{eq: appendix B prime definition} are closed under rescaled partial trace, in the sense that, given two numbers $\lambda_1,\lambda_2>0$ such that $\max_{B_1\in\sB_1}\tr B_1\leq\lambda_1$ and $\max_{B_2\in\sB_2}\tr B_2\leq\lambda_2$, if $R_{12}\in\sB'_{12}$, then
\begin{align}
    \lambda_2^{-1}\tr_2 R_{12}\in\sB'_1, \qquad \lambda_1^{-1}\tr_1 R_{12}\in\sB'_2.
\end{align}
\end{lemma}

\begin{proof}
We repeat the reasoning in the proof of Lemma~\ref{lem: partial trace normalized states appendix}. Since $\max_{B_2\in\sB_2}\tr B_2\leq\lambda_2$, we have $\lambda_2^{-1}I_2\in\polarPSD{(\sB_2)}$. For any $S_1\in\polarPSD{(\sB_1)}$, we thus deduce that
\begin{align}
    \lambda_2^{-1}S_1\ox I_2 \in \polarPSD{(\sB_1)}\ox\polarPSD{(\sB_2)} \subseteq \polarPSD{(\sB_{12})}.
\end{align}
Using Lemma~\ref{lem: appendix B prime polar}, this is also an element of $\polarPSD{(\sB'_{12})}$. Hence, for any $R_{12}\in\sB'_{12}$, we obtain
\begin{align}
    1 \geq \tr\left[R_{12}(\lambda_2^{-1}S_1\ox I_2)\right] = \lambda_2^{-1}\tr[R_1S_1],
\end{align}
where $R_1:=\tr_2R_{12}$. This entails that
\begin{align}
    \lambda_2^{-1}R_1 \in \left(\polarPSD{(\sB_1)}\right)^\circ\cap\PSD = \sB'_1,
\end{align}
where the last equality follows from Eq.~\eqref{eq: partial trace normalized bipolar appendix}. Exchanging the first and second systems gives the remaining inclusion.
\end{proof}

\section{Alternative proof for Lemma~\ref{lem: hypothesis operational and DH}}
\label{sec: appendix Alternative proof lem: hypothesis operational and DH}

Here we give an alternative proof for Lemma~\ref{lem: hypothesis operational and DH}, which follows the argument outlined in~\cite[Section III]{Brand_o_2010}. Denote
\begin{align}\label{eq: operational proof tmp3}
    f( \sA, \sigma, \ve) := \sup_{M}  \left\{\tr[\sigma M]: 0\leq M\leq I,\, \tr[\rho M] \leq \ve, \forall \rho \in \sA \right\}.
\end{align}
By using the definition of dual cone of a convex set, we get 
\begin{align}\label{eq: operational proof tmp1}
    f( \sA, \sigma, \ve) = \sup_{M} \left\{\tr[\sigma M]: 0\leq M\leq I, \ve I - M \in \sA^* \right\},
\end{align}
where $\sA^*:=\{X: \tr[X\rho] \geq 0, \forall \rho \in \sA\}$ is the dual cone of $\sA$. Then introducing the Lagrange multipliers $X \geq 0$, $Y \geq 0$ and $Z \in \Cone(\sA)$, with $\Cone(\sA)$ being the cone generated by $\sA$, we get the Lagrangian of $f( \sA, \sigma, \ve)$ as 
\begin{align}
    L(\sigma, \ve, M, X, Y, Z)  := & \tr[\sigma M] + \tr[XM] + \tr[Y(I-M)] + \tr[\ve I - M] Z\\
    = &  \tr[M(\sigma + X - Y - Z)] + \tr [Y] + \ve \tr[Z].
\end{align}
It is easy to see that $M = (\ve/2) I$ is a strictly feasible solution to the optimzation problem in Eq.~\eqref{eq: operational proof tmp1}. Therefore, by Slater's condition~\cite{boyd2004convex}, $f( \sA, \sigma, \ve)$ is equal to its dual program, which is given by
\begin{align}
    f(\sA, \sigma,\ve) = \inf_{Y, Z} \left\{\tr [Y] + \ve \tr[Z]: \sigma \leq Y + Z, Y \geq 0, Z \in \Cone(\sA)\right\}.
\end{align} 
Using that $\tr[V_+] = \min\left\{\tr[W]: W \geq 0, W \geq V\right\}$ with $V_+$ being the positive part of $V$, we have
\begin{align}
    f( \sA, \sigma, \ve) & = \inf_{Z} \left\{\tr [(\sigma - Z)_+] + \ve \tr[Z]: Z \in \Cone(\sA)\right\}\\
    & = \inf_{\rho \in \sA} \inf_{x > 0} \left\{\tr [(\sigma - x\rho )_+] + \ve x\right\},\label{eq: operational proof tmp2}
\end{align}
where the second line follows by introducing $Z = x \rho$ with $\rho \in \sA$ and the fact that $\sA \subseteq \density$ is a set of normalized quantum states. Applying the above result to the case that $\sA = \{\rho\}$, we get 
\begin{align}
    f(\rho,\sigma,\ve) = \sup_{M}  \left\{\tr[\sigma M]: 0\leq M\leq I,\, \tr[\rho M] \leq \ve \right\} = \inf_{x > 0} \left\{\tr [(\sigma - x\rho )_+] + \ve x\right\}.
\end{align}
Taking this to Eq.~\eqref{eq: operational proof tmp2}, we have
\begin{align}\label{eq: operational proof tmp4}
    f( \sA, \sigma, \ve) = \inf_{\rho \in \sA} \sup_{M}  \left\{\tr[\sigma M]: 0\leq M\leq I,\, \tr[\rho M] \leq \ve \right\}.
\end{align}
Replacing $M$ with $I-M$ in Eqs.~\eqref{eq: operational proof tmp3} and~\eqref{eq: operational proof tmp4}, we get
\begin{align}
    \inf_{M}  & \big\{\tr[\sigma M]:  0\leq M\leq I,\, \tr[\rho (I-M)] \leq \ve, \forall \rho \in \sA \big\} \\
    & = \sup_{\rho \in \sA} \inf_{M}  \left\{\tr[\sigma M]: 0\leq M\leq I,\, \tr[\rho (I-M)] \leq \ve \right\}.
\end{align}
In a simpler notation, this shows  
\begin{align}\label{eq: operational proof tmp5}
    \beta_\ve(\sA\|\{\sigma\}) = \sup_{\rho \in \sA} \beta_\ve(\rho\|\sigma).
\end{align}
Now we are ready to prove the asserted result. Note that both $\sB$ and $\{M:  0 \leq M \leq I, \alpha(\sA,M)\leq \ve\}$ are convex with the later being compact. Moreover, the objective function $\tr[\sigma M]$ is linear in both $\sigma$ and $M$. Therefore, we can apply Sion's minimax theorem~\cite[Corollary 3.3]{Sion1958} to exchange the infimum and supremum and get
\begin{align}
    \beta_\ve(\sA\|\sB) & = \inf_{\substack{0 \leq M \leq I\\\alpha(\sA,M)\leq \ve}} \sup_{\sigma \in \sB} \tr[\sigma M] = \sup_{\sigma \in \sB} \inf_{\substack{0 \leq M \leq I\\\alpha(\sA,M)\leq \ve}}  \tr[\sigma M].
\end{align}
Combining with Eq.~\eqref{eq: operational proof tmp5}, we have
\begin{align}
    \beta_\ve(\sA\|\sB)  & = \sup_{\sigma \in \sB} \beta_\ve(\sA\|\{\sigma\})  = \sup_{\sigma \in \sB} \sup_{\rho \in \sA} \beta_\ve(\rho\|\sigma).
\end{align}
This completes the proof.

\section{Useful properties}
\label{sec: app Useful properties}

We establish an explicit continuity for $D_{\alpha}$ was $\alpha \to 1$ with $\alpha < 1$. It is analogous to~\cite[Lemma 8]{tomamichel2009fully} and~\cite[Lemma B.8]{dupuis2020entropy} and the proof is almost the same.
\begin{lemma}[Explicit continuity of R\'enyi divergences as $\alpha \to 1$ from below]
\label{lem: continuity alpha = 1 from below}
Let $\rho \in \density(\cH)$ and $\sigma \in \PSD(\cH)$. Let $\eta = \left(\max(4, 2^{2 D_{\Petz,3/2}(\rho \| \sigma)} + 2^{-2D_{\Petz,1/2}(\rho \| \sigma)} + 1)\right)^2$. For $\alpha \in (1-1/\log \eta, 1)$, we have 
\begin{align}
    0 \leq D(\rho \| \sigma) - D_{\Petz, \alpha}(\rho \| \sigma) \leq (1-\alpha) (\log \eta)^2.
\end{align}
As a result, we also have
\begin{align}
    0 \leq D(\rho \| \sigma) - D_{\Sand, \alpha}(\rho \| \sigma) \leq (1-\alpha) \left((\log \eta)^2 + D(\rho \| \sigma) + \log \tr(\sigma) \right).
\end{align}
\end{lemma}
 \begin{proof}
 We show the result for classical distributions $P, Q$ over $\cX$ and then the result follows immediately for the Petz R\'enyi divergence by applying it to the corresponding Nussbaum-Szko{\l}a distributions~\cite{Nussbaum2009} as was done in~\cite{dupuis2019entropySecondOrder}. We follow the proof from~\cite[Lemma 8]{tomamichel2009fully}. Letting $\beta = \alpha - 1 < 0$, we have
 \begin{align}
 D_{\alpha}(P \| Q) 
 &= \frac{1}{\beta} \log \sum_{x \in \cX} P(x) \left(\frac{P(x)}{Q(x)}\right)^{\beta}.
 \end{align}
 Letting $t^{\beta} = 1 + \beta \ln t + r_{\beta}(t)$, where $r_{\beta}(t) = t^{\beta} - \beta \ln t - 1$, we can write
  \begin{align}
 D_{\alpha}(P \| Q) 
 &= \frac{1}{\beta} \log\left(1 + \beta \sum_{x \in \cX} P(x) \ln \frac{P(x)}{Q(x)}  + \sum_{x \in \cX} P(x) r_{\beta}\left(\frac{P(x)}{Q(x)}\right) \right).
 \end{align}
 We now use the fact that $\ln(1+y) \leq y$. Thus,
 \begin{align}
 D_{\alpha}(P \| Q) 
 &\geq \frac{1}{\beta} \frac{1}{\ln 2} \left(\beta \sum_{x \in \cX} P(x) \ln \frac{P(x)}{Q(x)}  + \sum_{x \in \cX} P(x) r_{\beta}\left(\frac{P(x)}{Q(x)}\right) \right) \\
 &= D(P \| Q) + \frac{1}{\beta \ln 2} \sum_{x \in \cX} P(x) r_{\beta}\left(\frac{P(x)}{Q(x)}\right).
 \end{align}
 We now define $s_{\beta}(t) = 2 (\cosh(\beta \ln t) - 1)$ as in~\cite{tomamichel2009fully}, but recall that we have $\beta < 0$. Note that $s_{\beta}(t) = s_{-\beta}(t)$ so we also have that it is monotonically increasing for $t > 1$ and concave in $t$ for $\beta > -1/2$ and $t \geq 3$. In addition, we have $r_{\beta}(t) = e^{\beta \ln t} + (-\beta \ln t + 1) - 2 \leq e^{\beta \ln t} + e^{-\beta \ln t} - 2 = s_{\beta}(t)$ for all $t \geq 0$. Note that we also have $s_{\beta}(t) = s_{\beta}(1/t)$ and $s_{\beta}(t^2) = s_{2\beta}(t)$. As a result, $s_{\beta}(t) \leq s_{2\beta}(1+\sqrt{t} + \sqrt{1/t})$. As a consequence, because we assumed $\alpha > 1-1/\log 16$, we have $2\beta \geq -1/2$ and $1+\sqrt{t}+\sqrt{1/t} \geq 3$ and we can use the concavity of $s_{2\beta}$ and get
  \begin{align}
 D_{\alpha}(P \| Q) 
 &\geq D(P \| Q) + \frac{1}{\beta \ln 2} \sum_{x \in \cX : P(x) > 0} P(x) s_{2\beta}\left(1 + \sqrt{\frac{P(x)}{Q(x)}} + \sqrt{\frac{Q(x)}{P(x)}}\right) \\
 &\geq D(P \| Q) + \frac{1}{\beta \ln 2} s_{2\beta}\left(\sum_{x \in \cX : P(x) > 0} P(x)\left(1 + \sqrt{\frac{P(x)}{Q(x)}} + \sqrt{\frac{Q(x)}{P(x)}}\right) \right) \\
 &\geq D(P \| Q) + \frac{1}{\beta \ln 2} s_{2\beta}\left(1 + 2^{2 D_{3/2}(P \| Q)} + 2^{-2 D_{1/2}(P \| Q)} \right) \\
 &\geq D(P \| Q) + \frac{1}{\beta \ln 2} s_{2\beta}\left(\sqrt{\eta} \right).
 \end{align}
 Note that by Taylor's theorem around $z = 0$, we have $e^{z} + e^{-z} - 1 \leq z^2 (e^z + e^{-z})$. As a result
 $\frac{1}{\beta \ln 2} s_{2\beta}(\sqrt{\eta}) \geq \frac{1}{\beta \ln 2} (2\beta \ln \sqrt{\eta})^2 \cosh(2\beta \ln \sqrt{\eta}) = \beta (\log \eta)^2 \ln 2 \cosh(\beta \ln \eta) \geq \beta (\log \eta)^2$ where we used the fact that $\ln 2 \cosh(\ln 2) < 1$, which proves that 
  \begin{align}
 D_{\alpha}(P \| Q) 
 &\geq D(P \| Q) - (1-\alpha) (\log \eta)^2.
 \end{align}

To prove the bound on the sandwiched R\'enyi divergence, we use the relation proved in~\cite[Corollary 2.3]{iten2016pretty} between the different divergences in the regime $\alpha \in (0,1)$:
\begin{align}
\label{eq: relation sand petz}
    D_{\Sand,\alpha}(\rho \| \sigma) \geq \alpha D_{\Petz,\alpha}(\rho \| \sigma) + (1-\alpha) (\log \tr(\rho) - \log \tr(\sigma)).
\end{align}
This implies 
\begin{align}
D_{\Sand,\alpha}(\rho \| \sigma) 
&\geq D_{\Petz, \alpha}(\rho \| \sigma) - (1-\alpha) D_{\Petz,\alpha}(\rho \| \sigma) - (1-\alpha) \log \tr(\sigma) \\
&\geq D(\rho \| \sigma) - (1-\alpha) \left((\log \eta)^2 + D(\rho \| \sigma) + \log \tr(\sigma) \right).
\end{align}
\end{proof}

For $\alpha > 1$, we use~\cite[Lemma 8]{tomamichel2009fully}, with a slightly different parametrization. The proof is the same as Lemma~\ref{lem: continuity alpha = 1 from below}.
\begin{lemma}
\label{lem: continuity alpha = 1 from above}
Let $\rho \in \density(\cH)$ and $\sigma \in \PSD(\cH)$.  For any $\alpha \in (1, 1+1/\log \eta )$, we have
\begin{align}
D_{\Sand, \alpha}(\rho \| \sigma) \leq D(\rho \| \sigma) + (\alpha - 1) \left( \log \eta \right)^2,
\end{align}
where $\eta = \left(\max(4, 2^{D_{\Petz,3/2}(\rho \| \sigma)} + 2^{-D_{\Petz,1/2}(\rho \| \sigma)} + 1)\right)^2$.
\end{lemma}

\end{document}